\documentclass[twocolumn,a4paper,longbibliography,pra,superscriptaddress,nobibnotes,floatfix,nofootinbib,raggedbottom]{revtex4-1}

\usepackage{xcolor}
\usepackage{amsmath}
\usepackage{amsthm}
\usepackage{amssymb}
\usepackage{braket}
\usepackage{latexsym}
\usepackage{bbm}
\usepackage{url}
\usepackage{graphicx}
\usepackage{algorithmicx}
\usepackage{algpseudocode}
\newcommand{\timestamp}[1]{\ensuremath{\text{TS}}} 
\newcommand{\ie}{\emph{i.e.}}
\newcommand{\eg}{\emph{e.g.}}

\newcommand{\vertex}[2]{\langle #1, #2\rangle}
\newcommand{\GF}[1]{\operatorname{GF}(#1)}
\newcommand{\nb}[1]{\operatorname{nb}(#1)}

\renewcommand{\epsilon}{\varepsilon}
\newcommand{\Eprint}[1]{\url{#1}}
\newcommand{\mdp}[2]{\ensuremath{\text{\texttt{#1}}.\text{\texttt{#2}}}}
\newcommand{\class}[1]{\ensuremath{\text{\texttt{#1}}}}
\newcommand{\call}[2]{\ensuremath{\text{\texttt{#1}}(#2)}}
\newcommand{\treq}{\ensuremath{t^{\text{req}}}}

\newcommand{\return}[1]{\textbf{return} \(#1\)}
\makeatletter\newcommand{\href@noop}{}
\ifdefined\@bibitemShut\relax\else\newcommand{\@bibitemShut}{}\fi\makeatother
\theoremstyle{remark}

\newcommand*{\cX}{\mathcal{X}}

\newcommand*{\cZ}{\mathcal{Z}}

\newcommand*{\eps}{\varepsilon}

\usepackage{ifthen}
\usepackage[shortcuts]{extdash} % for hyphenation of words after
\usepackage{hyperref}

\let\originalleft\left
\let\originalright\right
\renewcommand{\left}{\mathopen{}\mathclose\bgroup\originalleft}
\renewcommand{\right}{\aftergroup\egroup\originalright}

\theoremstyle{plain}
\newtheorem{thm}{Theorem}[section]  %[chapter] or [section] or nothing
\newtheorem{lem}[thm]{Lemma}

\theoremstyle{definition}
\newtheorem{deff}[thm]{Definition}

\theoremstyle{remark}

\renewcommand{\tensor}{\otimes}

\newcommand{\naturals}{\ensuremath{\mathbb{N}}}

\newcommand{\hconj}[1]{\ensuremath{#1^\dag}}

\DeclareMathOperator{\tr}{tr}
\newcommand{\strace}[2][@]{\ensuremath{\tr\ifthenelse{\equal{#1}{@}}{}{_{#1}}(#2)}}
\newcommand{\ltrace}[2][@]{\ensuremath{\tr\ifthenelse{\equal{#1}{@}}{}{_{#1}}\left(#2\right)}}
\newcommand{\ktrace}[2][@]{\ensuremath{\tr\ifthenelse{\equal{#1}{@}}{}{_{#1}}\left[#2\right]}}
\newcommand{\trace}[2][@]{\if@display\ltrace[#1]{#2}\else\strace[#1]{#2}\fi}

\newcommand{\HminOp}{H_{\min}}

\newcommand{\sHmin}[2][@]{\ensuremath{\HminOp(#2)\ifthenelse{\equal{#1}{@}}{}{_{#1}}}}
\newcommand{\lHmin}[2][@]{\ensuremath{\HminOp\left(#2\right)\ifthenelse{\equal{#1}{@}}{}{_{#1}}}}
\newcommand{\Hmin}[2][@]{\if@display\lHmin[#1]{#2}\else\sHmin[#1]{#2}\fi}

\newcommand{\snorm}[1]{\ensuremath{\|#1\|}}
\newcommand{\lnorm}[1]{\ensuremath{\left\|#1\right\|}}
\newcommand{\norm}[1]{\if@display\lnorm{#1}\else\snorm{#1}\fi}

\newcommand{\trnorm}[1]{\ensuremath{\norm{#1}_{\tr}}}

\DeclareMathOperator{\Ext}{Ext} % extractor
\newcommand{\coloneqq}{:=}

\begin{document}

% Less cumbersome float placement
\renewcommand{\topfraction}{0.9}	% max fraction of floats at top
\renewcommand{\bottomfraction}{0.8}	% max fraction of floats at bottom

\setcounter{topnumber}{2}
\setcounter{bottomnumber}{2}
\setcounter{totalnumber}{4}     % 2 may work better
\setcounter{dbltopnumber}{2}    % for 2-column pages
\renewcommand{\dbltopfraction}{0.9}	% fit big float above 2-col. text
\renewcommand{\textfraction}{0.07}	% allow minimal text w. figs
% Parameters for FLOAT pages (not text pages):
\renewcommand{\floatpagefraction}{0.7}	% require fuller float pages
% N.B.: floatpagefraction MUST be less than topfraction !!
\renewcommand{\dblfloatpagefraction}{0.7}	% require fuller float pages

\author{Wolfgang Mauerer}\email{wolfgang.mauerer@siemens.com
  (corresponding author)}
\affiliation{Siemens AG, Corporate Research and Technologies,
  Wladimirstra{\ss}e 1, 91058 Erlangen, Germany}
\author{Christopher Portmann}\email{chportma@phys.ethz.ch}
\affiliation{ETH Z{\"u}rich, Institute for Theoretical Physics,
  Wolfgang-Pauli-Stra{\ss}e 27, 8093 Z{\"u}rich}
\affiliation{Group of Applied Physics, University of Geneva, 1211 Geneva, Switzerland.}
\author{Volkher B.~Scholz}\email{scholz@phys.ethz.ch}
\affiliation{ETH Z{\"u}rich, Institute for Theoretical Physics,
  Wolfgang-Pauli-Stra{\ss}e 27, 8093 Z{\"u}rich}

\title{A modular framework for randomness extraction based on
  Trevisan's construction}

\begin{abstract}
  Informally, an extractor delivers perfect randomness from a source
  that may be far away from the uniform distribution, yet contains
  some randomness. This task is a crucial ingredient of any attempt to
  produce perfectly random numbers---required, for instance, by
  cryptographic protocols, numerical simulations, or randomised
  computations. Trevisan's extractor raised considerable theoretical
  interest not only because of its data parsimony compared to other
  constructions, but particularly because it is secure against quantum
  adversaries, making it applicable to quantum key distribution.

  We discuss a modular, extensible and high-performance implementation
  of the construction based on various building blocks that can be
  flexibly combined to satisfy the requirements of a wide range of
  scenarios. Besides quantitatively analysing the properties of many
  combinations in practical settings, we improve previous theoretical
  proofs, and give explicit results for non-asymptotic cases. The
  self-contained description does not assume familiarity with
  extractors.% and is accessible to non-specialists.
\end{abstract}
\maketitle

%%%%%%%%%%%%%%%%%%%%%%%%%%%%%%%%%%%%%%%%%%%%%%%%%%%%%%%%%%%%%%%%%%%%%%%%%%%%%%%%%%%%%%%%%%%%%
% Each section has its own tex file - that way we minimize the possibility to get collisions
% Writing duties:
% Section Intro and Overview: Volkher, Wolfgang
% Section on technical details: Christopher
% Section on implementation issues: Wolfgang
% Section on runtimes: Wolfgang / Christopher
%%%%%%%%%%%%%%%%%%%%%%%%%%%%%%%%%%%%%%%%%%%%%%%%%%%%%%%%%%%%%%%%%%%%%%%%%%%%%%%%%%%%%%%%%%%%%

\section{Introduction}
% Section Intro and Overview: Volkher
% should include an overview of the literatur, as well as a rough explanation of 
% the goal of our work and an overview of applications
% TODO: Use a consistent scheme for italics vs. regular
Random numbers are key ingredients for many purposes concerning
communication or computation: secretly shared, perfectly random bit strings
enable two parties to communicate in private using a one-time pad,
without the possibility of a third party decrypting any of the messages
they exchange.  Stochastic algorithms used in numerical simulation or
machine learning also rely on
random numbers as part of their input. In all such applications, it is
usually best to have \emph{uniform} randomness available, that is, an
observer should not have prior knowledge about the distribution of numbers,
or, more specifically, the content of bit strings: Each string should be
equally probable from his point of view. Some applications, such as
encrypting a message with information\-/theoretic security, are even
impossible if the randomness used to choose the key is not equivalent
to a uniform distribution~\cite{BD07}. Unfortunately, despite their
usefulness and the need for them, uniformly distributed random
bits are almost impossible to generate in practice. On the
other hand, there are plenty of physical resources containing ``some''
randomness, for instance radioactive decay, thermal fluctuations, certain
measurements on photons, or many others.

This contrast motivates the study of \emph{randomness extractors}:
Functions that map longer, slightly random bit strings onto shorter,
perfectly random bit strings.  They convert an initial distribution of
random numbers (the \emph{source}) that satisfies certain assumptions on
``how random'' it is into an almost uniform distribution over the output
bit strings. As suggested by intuition, this is impossible in a completely
deterministic way~\cite{Shaltiel2002}, and extractors indeed require a
second source of randomness, the \emph{seed}, that is usually assumed to
be perfectly uniformly distributed.

The goal of this work is twofold: to implement a specific
randomness extractor devised by Trevisan in 1999~\cite{Trevisan2001} as a practical
companion to the abundant amount of theoretical literature on the subject, and
to provide an overview and guidance on the topic to experimentalists who need
to use extractors, but would not benefit from working through all fundamental
publications. Trevisan's construction has three particular
advantages: For one, it is secure in the presence of quantum side information,
as was shown by one of the authors in collaboration with
others~\cite{De2012}. This is especially important in the context of
cryptography, where an adversary usually has some prior information about the
initial distribution used as raw material to produce a secret key. With a
quantum-proof extractor, it is possible to eliminate all these undesired
correlations by turning the initial distribution into a uniform one \--- a task
referred to as \emph{privacy amplification}. With quantum key
distribution (QKD) systems gradually
transitioning from research labs into commercial applications, it is very
important to implement this crucial protocol step, and given a bound on how
random and how correlated with some (quantum-)memory a bit string is, the
algorithm can indeed perform the task of producing truly random and uncorrelated
bits with the help of a short seed of uniformly distributed bits.

Another crucial advantage of Trevisan's construction is that the required
seed length is only poly-logarithmic in the length of the input. This
greatly outperforms randomness extractors based on (almost) universal hashing, which
are currently most often used in quantum cryptographic
applications~\cite{Ren05,TSSR11}, but require a seed whose size unfortunately
scales with the length of the raw input (output) bits.

In addition, Trevisan's construction is a \emph{strong} extractor,
which means that the seed is almost independent of the final
output. This implies that randomness in the seed is not consumed in
the process (compared to \emph{weak} extractors) and can be used at a
later time \--- or, as in the case of privacy amplification, it can be
obtained by the adversary without compromising the security of the QKD scheme.

Despite the considerable theoretical attention the field of extractors has
received during the last decade, there is, to our knowledge, only a single
publication, Ref.~\cite{Ma}, that discusses a prototypical implementation
of Trevisan's construction. However, their work has some drawbacks:
Compared to Ref.~\cite{Ma}, our implementation offers greater flexibility
as the operator can combine various different building blocks that make up
the extractor, and so can specifically engineer an algorithm for his
needs. Comparing the performance, our implementation exceeds the throughput
of~\cite{Ma} by several orders of magnitude, and is for the first time able
to scale to data sets of realistic size (exceeding the maximal amount
considered in~\cite{Ma} by 10 orders of magnitude) for which the amount of
extracted randomness actually exceeds the size of the initial seed, which
marks the regime in which Trevisan's construction prevails over
two-universal hashing. Besides, the full source code of our implementation
is available\footnote{The sources are available under the terms of the GNU
General Public License (GPL), version~2 \--- see
\url{www.gnu.org}. Essentially, this means that the code can be used and
modified free of charge for research (or even commercial) work, provided
that any improvements to the code are made available under similar terms.}
and can be inspected and used as basis for further research.  We therefore
hope that our implementation will be of use for applications in the context
of quantum cryptography, for implementing random number generators, or as a
testbed for developing new ideas about extractors.

% Up to now, we have deliberately skipped over several important issues. For
% example, we are usually satisfied with a final distribution that is not
% perfectly uniform, but only almost impossible to distinguish from a truly
% uniform one. We also did not specify any notion of entropy to quantify the
% amount of randomness contained in a semi-random source, and have also not
% considered conditional entropies to describe randomness relative to some
% observer.  Indeed, the familiar notion of Shannon's (or von Neumann's)
% entropy~\cite{Nielsen2010} is no longer sufficient in our context, and
% needs to be replaced by a more general one.

In Section~\ref{sec:overview}, we give more proper descriptions and 
definitions of the involved concepts and constructs. In particular, we
discuss the necessary notions of entropy and the distance of a
distribution from uniform (relative to an overserver). However, no
prior knowledge about randomness extractors is assumed. Section~\ref{sec:derivations}
contains the necessary technical details, and can be skipped upon
first reading. Section~\ref{sec:implementation} is devoted to the implementation: It
describes the software architecture and discusses some important technical
details, explains how to add new components, and gives concise algorithmic
descriptions of all components.  In Section~\ref{sec:runtime}, we present comprehensive
performance measurements, and discuss which combinations of primitives are
useful for which purpose.  The appendices collect formal definitions,
provide known extractor results with explicitly spelled out constants that
are, in contrast to many discussions that rely on asymptotic notations,
vital for an implementation, and give proofs for some new propositions
developed in the paper.

\section{Overview}
\label{sec:overview}
% Section Intro and Overview: Volkher, Wolfgang
% should explain in words what an extractor is, what the parameters are, explain 
% terminlogy like seed, min-entropy, side information
% give a rough overview of trevisans construction
% as well as explain (in words) what the possible choices of one-bit extractors are 
% good for
% Overview

\subsection{What are extractors?}

There are numerous possibilities to produce random numbers, and many
of them rely on some random physical process, turning, for instance,
thermal fluctuations into random bit strings. The laws of physics
state that these processes produce distributions with a non-zero
entropy, and hence are somewhat random. But the uniform distribution
or maximal entropy case is most often not within reach: thermal
fluctuations, for instance, require infinite temperatures to produce
truly random bit strings. It is therefore necessary to have an
algorithm that extracts random numbers from some given initial
distribution satisfying a lower bound on its entropy, turning them
into uniformly distributed ones. By shrinking the bit string (\ie,
reducing the support of its distribution), we increase its randomness
until it achieves its maximum.  It is easy to see that such a task is
impossible for any deterministic routine~\cite{Shaltiel2002}. But
assuming that we have two distributions (seed and source) over bit
strings at our disposal, promised to be uncorrelated and fulfilling a
lower bound on their entropy, the task comes into reach. Such
algorithms are called \emph{randomness extractors}, and their general
structure is shown in Figure~\ref{fig:scenario}.  The additional
randomness is usually taken to be uniform, and is called the
\emph{seed}\footnote{For simplicity we also treat the case of a
  uniform seed in this work, but some variations of Trevisan's
  extractor still apply when the additional randomness just fulfills a
  lower bound on its entropy~\cite{De2012}, and so the methods and
  code that we have developed can also be adapted to this setting.}. A
natural aim is to seek algorithms that minimise the required size of
the seed, or in other words, the amount of additional
randomness. Extractors depend on several parameters, specifying
source, seed, and output. This section explains the different
parameters and how they are quantified, and discusses their
connection.  In the second part, we briefly outline Trevisan's
construction.

\begin{figure}[htb]
  \includegraphics[width=\linewidth]{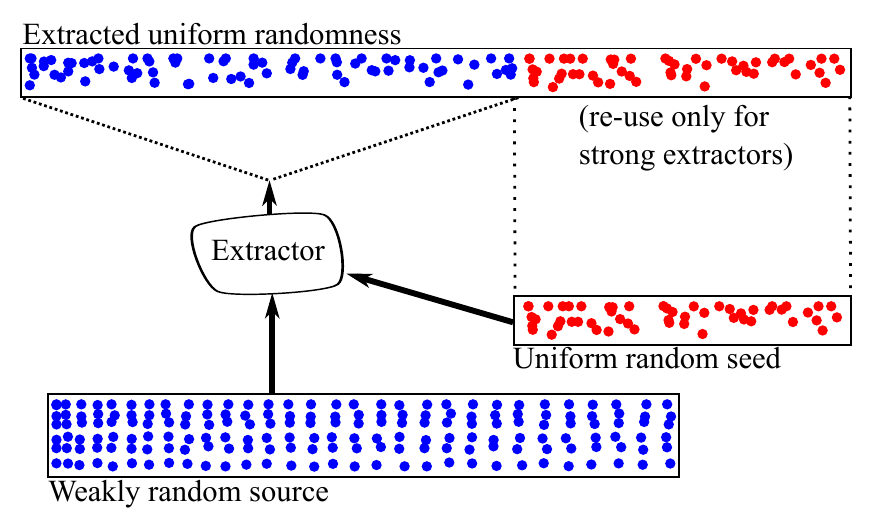}
 \caption{Scenario considered in randomness extraction: The output of
    a weakly random source, together with a short, uniformly
    distributed seed, is processed by a deterministic algorithm \---
    the \emph{extractor} \--- that produces a uniform sequence of
    randomness that is shorter than the initial input. For strong
    extractors, the initial seed is independent of the output and can
    be re-used as part of the result.}
  \label{fig:scenario}
\end{figure}

First, let us consider how to quantify the amount of randomness contained
in the source. As per the seminal works of Boltzmann, Shannon and von
Neumann, the amount of randomness contained in some distribution of numbers
is best quantified by its entropy, traditionally given as $\sum_x p_x \log
p_x$. Here, $x$ ranges over all output values, and $p_x$ is the probability
to observe outcome $x$. This notion of entropy originates from statistical
mechanics, where we deal with large numbers of independent entities that
are usually also identically distributed. In contrast to that, we are
interested in a \emph{single} run of our extractor, and \emph{not} in
statements about the output distribution obtained from many instances of
the extractor applied to many independent copies of the initial
distribution. Consequently, we have
to alter the notion of entropy.

Intuitively, the amount of randomness in some distribution is quantified by
the ability to predict the observed values. This leads to the definition of
the \emph{guessing probability} $p_{\text{guess}}(X)$ as the probability of
correctly guessing the value of the random variable $X$. It is given by
$p_{\text{guess}}(X) = \max_x p_x$ \--- the optimal strategy is to guess
the most probable value. The bigger $p_{\text{guess}}$, the
less random the source is. This is quantified by the \emph{min-entropy},
defined by $\Hmin{X} = - \log \max_x p_x$.

The definition does so far not consider the possible presence of side
information. In a more complex setting, there might be some side
information $E$ correlated to the source $X$, and the task becomes to
extract uniform randomness from $X$ that is independent of $E$. In a
cryptographic context, $E$ represents the adversary's information about
the source. Clearly, if $E$ is a one-to-one copy of $X$, this task is
impossible, even if $X$ is perfectly uniform. The notions of guessing
probability and entropy consequently need to be extended such that they
measure the randomness of the source \emph{conditioned on $E$}.  If the
side information is classical, then extractors proven sound in the
absence of side information can be used with only a small adjustment of the
parameters.\footnote{See Lemma~\ref{lem:classicalproof} for an exact
  statement.} However, this changes
dramatically if the observer is allowed to use the power of quantum
mechanics~\cite{Kempe}.

To guess the value of $X$, a player holding a quantum state in a
system $E$ may measure this system, and make a guess based on the
observed outcome. For every value $X=x$, his quantum memory is in some
conditional state $\rho_x$, and his task reduces to distinguishing the
different states $\rho_x$. Mathematically, such a measurement is
specified by a positive operator-valued measurement $\{E_x\}$. Thus,
the probability to correctly guess the value taken by $X$ is given by
$p_{\text{guess}}(X|E) = \sum_x p_x \tr \rho_x E_x$. The corresponding
entropic quantity, the conditional min-entropy~\cite{Ren05}, is given
by $\Hmin{X|E} = - \log \sum_x p_x \tr \rho_x E_x$, where we take
$\{E_x\}$ to be the optimal measurement.

Having specified the quantification of randomness, we need to define
what we mean by an ``almost uniform'' distribution over the output $Z
= \Ext(X,Y)$, where $X$ is the source and $Y$ the seed. Again
intuitively, we would like to assure that a player holding some side
information cannot do better than with a random guess, that is, the
probability that he guesses correctly should be close to
$\frac{1}{2^m}$ if the output is a bit string of length
$m$. Mathematically, this is expressed by requiring that the joint
state of the output and the side information $\rho_{ZE}$ is close to a
product state of a perfectly uniform output, $\tau$ \--- the fully
mixed state \--- and the side information $\rho_E$, that is, we want
$\rho_{ZE} \approx \tau \tensor \rho_E$. The distance\footnote{We use
  the trace distance to measure how close two states are, see
  Appendix~\ref{app:extractordefs} for an exact definition.} between
these states is usually denoted by $\eps$, and referred to as the
error of the extractor. Colloquially, an error of $\eps$ corresponds
to a probability of at most $\frac{1}{2^m} + \eps$ that the output can
be guessed correctly, and a probability of at most $\frac{1}{2} +
\eps$ that any single bit can be guessed.

We are now able to define extractors in more detail. We assume that the
input are bit strings of length $n$ and that the distribution has a
conditional min-entropy of at least $k$. For processing each input
string, $d$ randomly distributed bits may be used. The output should
consist of bit strings having length $m$, and the distribution of outcomes
should be $\eps$-close to uniform and independent from the side
information. We call a deterministic function taking as input the source
and the seed and achieving these goals a \emph{quantum-proof
  $(k,\eps)$-extractor}\footnote{See Definition~\ref{def:quantumextractor}
  for a formal definition.}. The \emph{output length} of such an extractor is
$m$. Naturally, we would like to have $m$ as close to $k$ as possible,
which means that most of the entropy has been extracted. The value $k-m$ is
therefore called the \emph{entropy loss}. The extractor is called
\emph{strong} if the output is also close to independent of the value of
the seed, or equivalently, the output of the extractor is a pair of bit
strings, the first being the value of the random bits used as seed, and the
second being the output. This is exactly the setting needed for the privacy
amplification step in quantum key distribution protocols, as both Alice and
Bob need to use the same value for the seed, in order to produce a
correlated bit string. It is thus assumed that the bit values for the seed
are uniformly distributed, but known to the adversary, since they are
publicly announced by one of the parties.

The most commonly used (at least in theoretical considerations) strong
extractor in quantum key distribution protocols is based on
two-universal hash functions~\cite{Ren05,TSSR11}. A \emph{family} is a
collection of functions that map longer bit strings to shorter bit
strings. Over a random choice of the function from the set,
two-universality requires that it is extremely unlikely for different
bit strings to be mapped to the same output. While universal hashing
is optimal in the entropy loss\footnote{An extractor will always have
  an entropy loss $\Delta \geq 2 \log 1/\eps + O(1)$, where $\eps$ is
  the error of the extractor~\cite{RT00}.}, the required seed length (the size of the
function family: as many bits as necessary to randomly select one
member) scales as a multiple of $n$, the input data length (or, in
the case of \emph{almost} universal hashing~\cite{TSSR11}, as a
multiple of $m$, the output length).

It is important to emphasise that strong extractors provide security
just based on an entropic assumption, namely the amount of
(conditional) min-entropy of the initial distribution. In contrast,
pseudo-random number generators are based on complexity theoretic
assumptions. For instance, the presumed existence of functions that
are hard to invert on average in polynomial time can be turned into an
algorithm taking a short random seed and producing an output
distribution that ``looks'' like the uniform distribution to any
algorithm running in polynomial time (see Ref.~\cite{HILL} for further
information and formal definitions). While such generators greatly
outperform our current implementation,\footnote{Practical
  implementations of pseudo-random number generators, among them the
  variant used in the Linux kernel~\cite{Mauerer2008}, rely on
  cryptographic hash functions like SHA-512~\cite{FIPS180}.  Since
  these functions, in turn, are used in numerous computing scenarios
  that extend well beyond cryptography, many recent CPUs offer
  special-purpose machine instructions that allow for particularly
  efficient implementations. This makes it practically impossible for
  an implementation of Trevisan's construction to beat the throughput
  of cryptographic hash algorithms that are, besides, much simpler
  from an algorithmic point of view.} they require much stronger
assumptions and give rise to weaker promises on the output
distribution.

After this general discussion on extractors and related issues, we now
describe Trevisan's construction in more detail.

% Prior to \cite{Trevisan}, the only constructions to be known to
% fulfill the requirements of being a strong extractors secure against
% quantum side information were two-universal hash functions
% \cite{RennerKoenig}. These are families of mappings on bit strings,
% where over a random choice of the function it is (extremely)
% unlikely that different bit strings are mapped to the same
% output. While universal hashing has entropy loss equal to zero, the
% required seed length (the size of the function family) is rather
% unfortunate, scaling like multiple of $n$. In contrast, Trevisan's
% construction \cite{Luca} only requires polylogarithmic seed length
% and was shown to be a strong quantum proof extractor in
% \cite{Trevisan}.

\subsection{Trevisan's Construction}

Trevisan's seminal contribution originates in the insight that a certain
class of error-correcting codes (ECC), called \emph{list-decodable codes}~\cite{Vadhanreview},
can be re-interpreted as
extractors. In fact, the codes are \emph{one-bit} extractors, and deliver a
single perfectly random bit from a larger reservoir of slightly random
bits. Since an error correcting code is a deterministic mapping from
shorter into longer bit strings to make them more robust against the
influence of errors acting on the encoded data, the connection between ECCs
and bit extractors is not immediately obvious. Trevisan's first observation
was that if we randomly select a position of an ECC's output string, the
corresponding bit is uniformly distributed, provided that the initial
distribution has enough min-entropy. If the code outputs bit strings of
length $\bar{n} = \text{poly}(n)$, a logarithmic long seed of random bits
is needed, since exactly $\log \bar{n}$ bits are necessary to specify a position
of an $n$-bit string.

Of course, we are interested in much longer outputs than just a single bit.
The second observation of Trevisan states that outputs of many uses of the
one-bit extractor can be concatenated so that the output is still uniformly
distributed, and that we do not need to choose a completely new set of
random seed bits for every use of the one-bit extractor. This is achieved
using a construction of Nisan and Wigderson~\cite{NW}, the Nisan-Wigderson
pseudo-random generator. The basic idea is that the initial choice of random
bits taken from the seed is divided into sets of random bits with small
overlap. For example, 100 random bits are divided into 15 sets, each
consisting of 10 bits. If the overlap is not too large, there are not too
many correlations induced by seeding the elements of each set into the
one-bit extractors and concatenating the output bits. The randomness
available in the initial distribution can then be used to cope with these
additional correlations. Dividing the original seed bits into smaller sets
is done using an algorithm called \emph{weak design}. The complete process
is summarised in Figure~\ref{fig:trevisan_extractor}.

\begin{figure}[htb]
  \includegraphics[width=\linewidth]{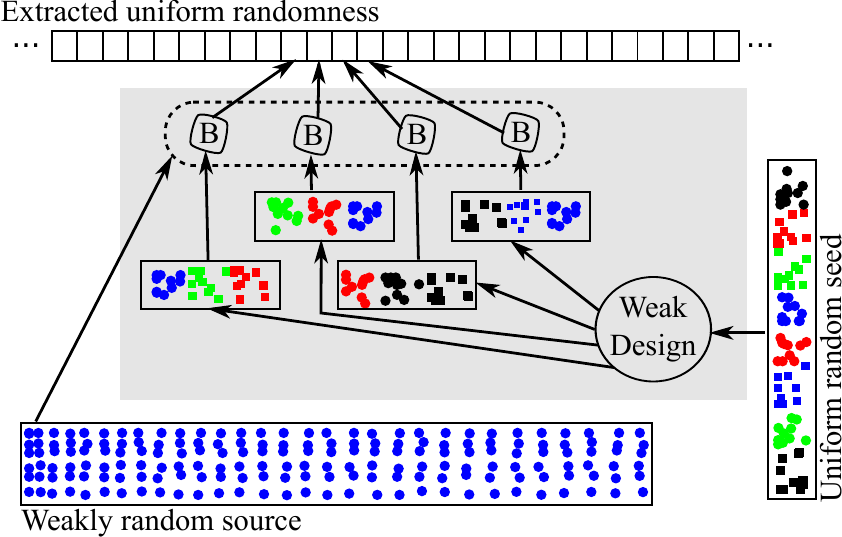}
  \caption{Interplay of components in Trevisan's extractor. The weak
    design expands the initial seed into multiple smaller packets with
    certain overlap properties whose cumulative length can
    considerably exceed the seed length. Each is fed into a 1-bit
    extractor \(B\), which distills a single random bit out of the
    global randomness for each packet. The bits are finally
    concatenated to form the extracted randomness. All components that
    make up Trevisan's algorithm are highlighted in a grey box.}
  \label{fig:trevisan_extractor}
\end{figure}

It turns out that there are many examples of one-bit extractors and
weak designs that fulfil the requirements needed for the above
procedure to work. Trevisan's construction is therefore not really a
single algorithm, but rather a recipe to combine different one-bit
extractors and weak designs to generate a quantum-proof strong
extractor. The exact choice of either building block (we also refer to
them as \emph{primitives} in the following) depends very much on the
application and on the parameter regime of interest. Consequently, we
decided to implement different possible choices and let the operator
decide which ones to use. We now present two exemplary use cases that do
especially highlight the need to prioritise between speed, entropy
loss, and the assumptions on the initial randomness.

Suppose first that we have at our disposal a fast source providing
very good random numbers, or equivalently, having a very high
entropy. Ideally, we would like to extract all randomness, but since
producing new random numbers is fast, we can allow a substantial
entropy loss, concentrating on performance instead. In this case, the
combination of the GF(\(p\))-weak design with the XOR one-bit
extractor is the best choice, achieving a throughput of about
17~kbit/s on a normal notebook machine and about 160~kbit/s on a large
% NOTE: The performance figure for the macbook is not contained
% in any other graph, but was computed specifically for the intro.
workstation with 48~cores. The extractor can handle input lengths of
several GiBits, which is also necessary since in these extreme cases
only one percent of the available entropy is extracted. This means
that the source needs to provide random numbers with a rate of about
20~Mibit/s. The required amount of seed for the one-bit extractor is
1.7~KiBit, which leads to a total seed of roughly 2.9~MiBit for
4~GiBit of input data.
% NOTE: Seed computed with compute.n.xor(4*10**9, 0.95, 0.01, 10**-7)

If we consider a source of very low entropy and focus on small
entropy loss rather than throughput, the optimal choice turns out to
be the block weak design together with polynomial hashing for the one-bit
extractor. It works for any lower bound on the entropy, has almost
minimal entropy loss, and requires the shortest seed of all
constructions. It is, however, much slower than the first combination:
A throughput of only a few kbit/s is achieved on a notebook computer,
or 70~kbit/s with 48~cores, albeit for a much shorter input length of
\(2^{16}\) bits: 100 bits are necessary for the one-bit extractor,
which results in 10~KiBit of total seed for the standard weak
design, and slightly less than 300~KiBit for the block weak design
needed to extract nearly all the entropy.
% NOTE: Seed computed with compute.n.rsh(2**16, 0.9, 1, 10**-7)
% and l.block(2*exp(1), 2**15, 100) to determine the number of blocks.

These are just two examples, and proper performance measurements as well as
a discussion on possible improvements and aspects of high-performance
computing can be found in section~\ref{sec:runtime}.

%%% Local Variables:
%%% TeX-master: "../trevisan.tex"
%%% End:

\section{Derivations}
\label{sec:derivations}
% Section on technical details: Christopher
% Explain the implemenation of the weak design, discuss trevisan's construction
% discuss various one-bit extractors, derive seed lengths, discuss the 
% possible choices

\subsection{Trevisan's extractor}
\label{sec:derivations.trevisan}

\subsubsection{Description}

As briefly sketched in the previous section, Trevisan's construction
consists in applying multiple times the same one-bit extractor to the
input string, using different weakly correlated seeds for each
run. The seeds are chosen as substrings of some longer seed $y \in
\{0,1\}^d$. Let $\{S_i\}_i$ be a family of sets such that for all $i$,
$|S_i| = t$ and $S_i \subset [d] = \{1,\dotsc,d\}$. Then $y_{S_i}$
\--- the string formed by the bits of $y$ at the positions given by
the elements $j \in S_i$ \--- is a string of length $t$. For a given
one-bit extractor $C : \{0,1\}^n \times \{0,1\}^t \to \{0,1\}$, and
such a family of sets $\{S_i\}_{i = 1}^m$, Trevisan's extractor is
defined as the concatenation of the output bits of $C$ when used with
the seeds $y_{S_i}$, namely
 \begin{equation} \label{eq:trevisansconstruction} \Ext(x,y) \coloneqq
   C(x,y_{S_1}) \dotsb C(x,y_{S_m}).\end{equation}

The performance of the extractor naturally depends on the performance
of the one-bit extractor, but also on the independence of the seeds
used for each run of the one-bit extractor. Intuitively, the smaller
the cardinality of the intersections of the sets
$\{S_i\}$, the more randomness we can extract form the source, but the
larger the seed. The exact condition is given in the following
definition.

\begin{deff}[weak design~\protect{\cite{Raz2002}\footnote{The second
      condition of the weak design was originally defined as $\sum_{j
        = 1}^{i-1} 2^{|S_j \cap S_i|} \leq r(m-1)$. We prefer to use
      the version of \cite{Hartman2003}, since it simplifies the notation
      without changing the design constructions.}}]
  \label{def:weakdesign} A family of sets $S_1,\dotsc,S_m \subset [d]$
  is a \emph{weak $(m,t,r,d)$-design} if
  \begin{enumerate}
    \item For all $i$, $|S_i| = t$.
    \item For all $i$,  $\sum_{j = 1}^{i-1} 2^{|S_j \cap S_i|} \leq
      rm$.
  \end{enumerate}  
\end{deff}
In the following, we refer to the parameter $r$ as the \emph{overlap}
of the weak design.

As an example, if we use a quantum-proof $(k,\eps)$-strong extractor
as one-bit extractor and a weak $(m,t,r,d)$-design, the construction
given by \eqref{eq:trevisansconstruction} is a quantum-proof
$(k+rm,m\eps)$-strong extractor (see Lemma~\ref{lem:implicit}). Thus,
if $r = 1$, Trevisan's extractor has roughly the same entropy loss as
the underlying one-bit extractor. Note also that the error $\eps$ of
the one-bit extractor is the error per bit for Trevisan's construction.

\subsubsection{Constructions overview}\label{sec:primitives_overview}

We always denote the input length by $n$, and the output length by
$m$. We choose $\eps$ in such a way that it corresponds to the error
per bit for the final construction. We use $d$ to describe the seed
length of Trevisan's extractor and $t$ for the seed length of the
underlying one-bit extractor. $r$ denotes the overlap of the weak
design, and $k$ the min-entropy required in the source, which is often
expressed as $k = \alpha n$.

In the following, we briefly summarize the constructions described in
Sections~\ref{sec:derivations.weakdesigns} and
\ref{sec:derivations.extractors}. We take the input length $n$, output
length $m$, and error per bit $\eps$ to be fixed, and calculate the
seed length $d$ and entropy needed in the source $k$ as functions of
these three parameters.

\paragraph{Weak designs:} In Section~\ref{sec:derivations.weakdesigns}
we describe two weak designs, the first was originally proposed by
Nisan and Wigderson~\cite{NW94}, and has parameters $d = t^2$ and $r =
2e$ for any prime power $t$ and any $m$. This means that the seed of
the final construction is the square of the seed of the one-bit
extractor, and the entropy loss induced by the weak design is
$(2e-1)m \approx 4.436 m$.  The second construction iterates the
first; it has a larger $d = a t^2$
for \begin{equation} \label{eq:overview.blocdesign} a = \left\lceil
    \frac{\log(m-2e) - \log(t-2e)}{\log 2e - \log(2e-1)}
  \right\rceil, \end{equation} but $r =1$, \ie,
the design does not cause any entropy loss.

\paragraph{XOR-code:} The XOR-code is a one-bit extractor, which
simply computes the XOR of a substring of the input. With the two
different weak designs, we find that the randomness and seed needed
are
\begin{align*}
  k & = \gamma n + rm + 6 \log \frac{1+\sqrt{2}}{\eps} + \log \frac{4}{3}, \\
  t & = \frac{2 \ln 2}{h^{-1}(\gamma)} \log n \log \frac{(2+\sqrt{2})}{\eps}, \\
  d & = t^2 \text{ or } at^2,
\end{align*}
where $\gamma$ is a free parameter that influences the amount of
extracted randomness and the length of the initial seed (details in
Section~\ref{sec:derivations.extractors.xor}), and $a$ is given by
Eq.~\eqref{eq:overview.blocdesign}. $h(p) = -p \log p - (1-p) \log
(1-p)$ is the binary entropy function, and its inverse is defined on
the interval $(0,1/2)$.

\paragraph{Lu's construction:} This one-bit extractor selects a random
substring of the input by performing a walk on an expander graph, and
then hashes the result to one bit by taking the parity of the bitwise
product with a random string. With the two different weak designs, we
find that the randomness and seed needed are
\begin{align*}
  k & = h(\nu)n + rm + 6 \log \frac{2+\sqrt{2}}{\eps} - 2, \\
  t & = \log n + 3 c(\ell-1) + \ell, \\
  c & = \left\lceil \frac{\log w}{2 \log 5\sqrt{2}/8} \right\rceil, \\
  \ell & = \left\lceil \frac{8 \log \eps - 8 \log (2+\sqrt{2})}{\log
    (1-\nu+w)} \right\rceil, \\
  d & = t^2 \text{ or } at^2,
\end{align*}
where $\nu \leq 1/2$ is a free parameter, $a$ is given by
Eq.~\eqref{eq:overview.blocdesign}, and $w$ is the solution to the
equation\footnote{$w < \nu$ can actually be chosen freely. The above
  value minimizes the walks on the expander graph.} \[w \log w =
(1-\nu+w)\log (1-\nu+w). \]

\paragraph{Polynomial hashing:} This constructions uses almost
universal hash functions. With the two different weak designs, we
find that the randomness and seed needed are 
\begin{align*}
  k & = rm + 4 \log \frac{1}{\eps}+6, \\
  t & = 2 \left\lceil \log n + 2 \log \frac{2}{\eps} \right\rceil,\\
  d & = t^2 \text{ or } at^2,
\end{align*}
where $a$ is given by \eqref{eq:overview.blocdesign}.

\subsection{Weak designs}
\label{sec:derivations.weakdesigns}

The weak design construction we use (see
Section~\ref{sec:derivations.weakdesigns.basic} for a description) is
originally from Nisan and Wigderson~\cite{NW94}, who proved that it is
a standard design \--- a notion stronger than \emph{weak designs},
originally used by Trevisan~\cite{Trevisan2001}, but which Raz et
al.~\cite{Raz2002} showed to be unnecessary. Hartman and
Raz~\cite{Hartman2003} proved that this construction is a weak
$(m,t,r,d)$-design with overlap $r = e^2$ for a prime $t$, $d=t^2$,
and $m$ a power of $t$. Ma and Tan~\cite{MT11} improved Hartman and
Raz's analysis, and showed that $r = e$ for any prime power $t$ and
any $m$ which is a multiple of a power of $t$. However, for a
practical implementation, we need a construction that is valid for any
$m$. We prove in Appendix~\ref{app:weakdesignproofs.basic} that this
construction is a weak $(m,t,r,d)$-design for any prime power $t$, any
$m$, and $r = 2e$.\footnote{Hartman and Raz~\cite[Corollary
  2]{Hartman2003} show that there exist a $d = O(t^2)$ and $r = O(1)$
  such that for any $m > t^{\log t}$ the construction is a
  $(m,t,r,d)$-design, however the restriction $m > t^{\log t}$ and
  constants in the $O$-notation which depend on $m$ make this unusable
  in practice. Ma and Tan~\cite{MT11} conjecture that the basic
  construction is a weak $(m,t,e,t^2)$-design for any $m$, and use
  this in their implementation. To make up for the lack of proof, they
  simply count the intersections between the sets $S_i$ after
  generating the design, to make sure that the overlap is bounded by
  $e$.}

As mentioned in Section~\ref{sec:derivations.trevisan}, a larger
overlap leads to a larger entropy loss. In
Section~\ref{sec:derivations.weakdesigns.reducing} we adapt an
iterative construction of the basic design from Ma and
Tan~\cite{MT11}, to construct a new design with $r = 1$. We prove in
Appendix~\ref{app:weakdesignproofs.reducing} that this construction is
correct.

\subsubsection{Basic construction}
\label{sec:derivations.weakdesigns.basic}

In this section we describe a weak design construction, that is, we
define a family of sets that satisfy the conditions of
Definition~\ref{def:weakdesign}.

This construction makes use of polynomials over a finite field
$\GF{t}$. Every set $S_p$ is indexed by one such polynomial $p :
\GF{t} \to \GF{t}$. To construct a weak $(m,t,r,d)$-design we need $m$
sets, and therefore $m$ such polynomials, which we take in increasing
order of their coefficients. For example, if $m = 6$ and $t = 2$, the
polynomials are $\sum_{i=0}^2 \alpha_i x^i$, with the coefficients
\((\alpha_{2}, \alpha_{1}, \alpha_{0})\) taken in the following order:
\((0,0,0)\), \((0,0,1)\), \((0,1,0)\), \((0,1,1)\), \((1,0,0)\),
\((1,0,1)\).  In general, the $n$th polynomial is given by $p(x) =
\sum_{i = 0}^{c} \alpha_i x^i$, with $\alpha_i = (n-1)/t^i \mod t$ and
$c = \left\lceil \frac{\log m}{\log t} - 1\right\rceil$.

The elements of the set $S_p$ are all the pairs of values $S_p
\coloneqq \{(x,p(x)) : x \in \GF{t}\}$. Each set thus has $|S_p| = t$
elements, and $S_p \subset [d]$ holds for $d = t^2$, where we
map $[d]$ to $[t] \times [t]$ in the obvious way. We prove in
Appendix~\ref{app:weakdesignproofs.basic} that for all $m$ and $p$,
this construction has $\sum_{\{q < p\}} 2^{|S_p \cap S_q|} \leq 2em$,
where by $\{q < p\}$ we mean the set of all polynomials that come
before $p$.

\subsubsection{Reducing the overlap}
\label{sec:derivations.weakdesigns.reducing}

Note that any weak design can be viewed as a binary $(m \times
d)$-matrix $W$, where the value $w_{ij} = 1$ if $j \in S_i$. To
construct a weak design with $r = 1$, we will use the construction
from Section~\ref{sec:derivations.weakdesigns.basic} repeatedly with
different values $m_j$ (but the same $t$), obtaining designs
$W_{\text{B}, 0},\dotsc,W_{\text{B}, \ell}$. We then construct a new
design $W$ by placing these in its diagonal, that is, \[ W
= \left( \begin{array}{ccc} W_{\text{B}, 0} & & \\
    & \ddots & \\
    & & W_{\text{B}, \ell} \end{array} \right).\]

Let $m$ and $t$ be fixed, and let $r' = 2e$ be the parameter from the
basic construction. The number of calls to the basic construction is
given by 
\begin{equation} \label{eq:r.ell}
  \ell \coloneqq \max \left\{1,\left\lceil \frac{\log(m-r') -
        \log(t-r')}{\log r' - \log(r'-1)} \right\rceil\right\}.
\end{equation}
And each design $W_{\text{B}, i}$ is constructed with $m_i$ sets,
defined as follows:
\begin{equation} \label{eq:r.m} \begin{split}
  n_i & \coloneqq \left(1 - \frac{1}{r'}\right)^{i} \left(\frac{m}{r'} - 1\right) \qquad \text{for
    $0 \leq i \leq \ell -1$}, \\
  m_i & \coloneqq \left\lceil\sum_{j = 0}^{i} n_j \right\rceil
  - \sum_{j = 0}^{i - 1} m_j \qquad \text{for $0 \leq i \leq \ell -1$}, \\
  m_\ell & \coloneqq m - \sum_{j = 0}^{\ell - 1} m_j.
\end{split} \end{equation}

The weak design $W$ thus has $d = (\ell+1)t^2$. In
Appendix~\ref{app:weakdesignproofs.reducing} we prove that this
construction has $r = 1$.

Figure~\ref{fig:blockdes} discusses the parameter behaviour of the
block weak design.

\begin{figure}[htb]
 \centering\includegraphics[width=\linewidth]{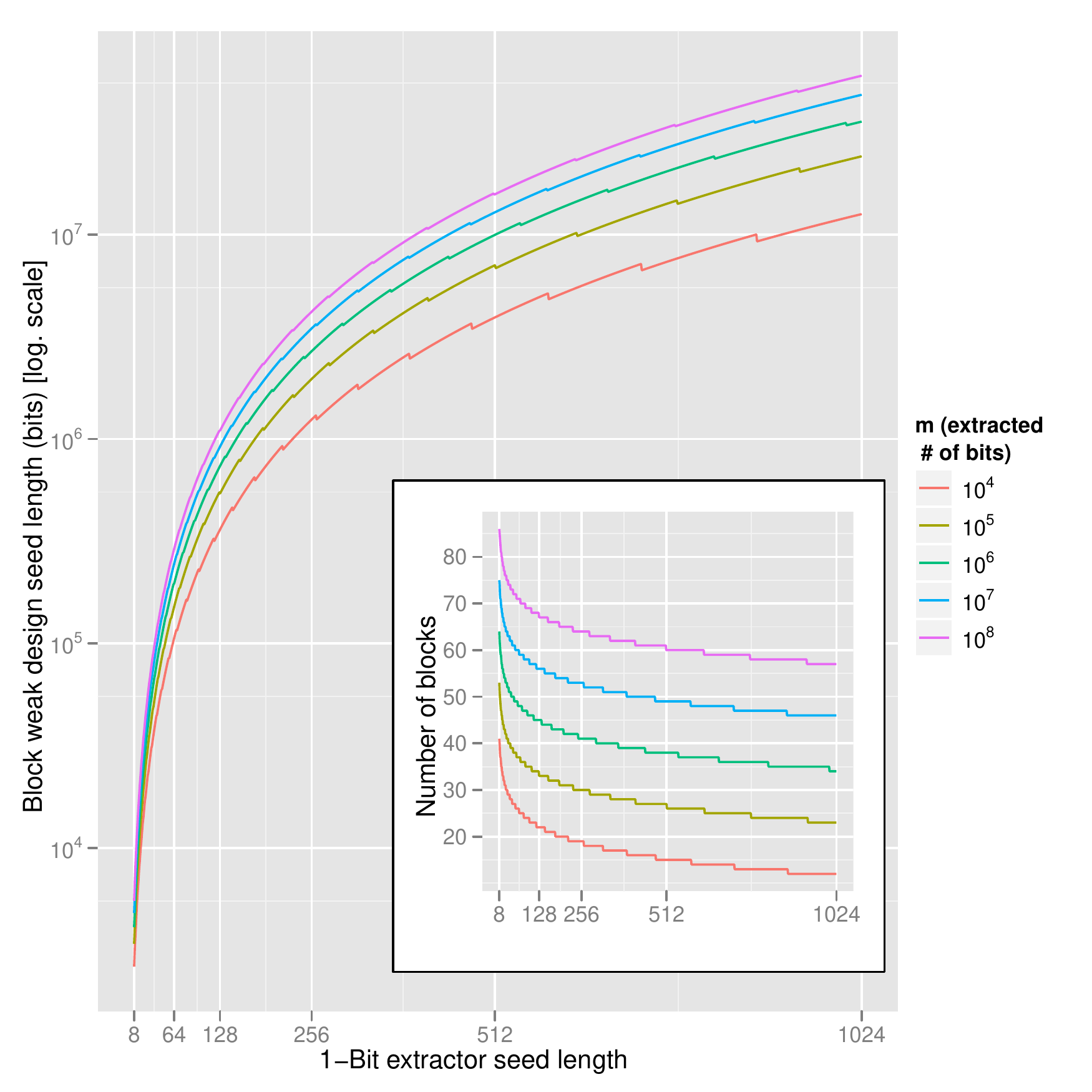}
 \caption{Total seed length required for the weak block design (based
   on an elementary design with \(r=2e\)) depending on the seed length
   of the 1-bit extractor. The inset shows the number of blocks --
   since the seed size for the block design depends linearly on the
   block number, this is also the seed overhead associated with the
   use of a block design. For typical parameters, the amount of seed
   grows by a factor of \(50\).}\label{fig:blockdes}
\end{figure}

\subsection{One-bit extractors}
\label{sec:derivations.extractors}

\subsubsection{XOR-code}\label{sec:derivations.extractors.xor}

This extractor computes the XOR for $\ell$ random positions of the
input, it is thus an $\ell$-local extractor (see
Appendix~\ref{app:extractordefs} for a precise definition). This
construction is efficient to compute, but requires a seed of length $t
\in O(\log n \log 1/\eps))$, where $n$ is the input length and $\eps$
the error of the construction, instead of the optimal $O(\log n + \log
1/\eps)$. It also has an entropy loss linear in the input length.

\begin{lem}[XOR-code \protect{\cite[Theorem
    41]{IJK09}}\footnote{\cite[Theorem 41]{IJK09} actually proves that
    this construction is a $\delta$-approximately
    $(\eps,L)$-list-decodable code. But such a code is an
    $(\eps,L2^{h(\delta)n})$-list-decodable code, which in turn is a
    classical\-/proof extractor by
    Lemma~\ref{lem:codesRextractors.bis}.}]
  \label{lem:localextractor}
  For any $\eps > 0$, $n \in \naturals$ and $\ell \in [n]$, the
  function
  \begin{align*}
    C_{n,\eps,\ell} : \ & \{0,1\}^n  \times [n]^\ell \to \{0,1\} \\
     & (x,i_1,\dots,i_\ell) \mapsto \bigoplus_{j = 1}^\ell x_{i_j}
  \end{align*}
  is a classical\-/proof $\ell$-local $(k,\eps)$-strong extractor with
  $k = h\left(\frac{\ln 2}{\ell} \log \frac{2}{\eps} \right) n + 3
  \log \frac{1}{\eps} +\log \frac{4}{3}$ and seed length $t = \ell
  \log n$, where $h(p) = -p \log p - (1-p) \log (1-p)$ is the binary
  entropy function.
\end{lem}

By Lemma~\ref{lem:1-bit-against-Q.bis}, this construction is a
quantum\-/proof $(k,(1+\sqrt{2})\sqrt{\eps})$-strong extractor. And by
Lemma~\ref{lem:implicit}, if we use this in Trevisan's construction, the
final extractor is a quantum\-/proof $(k+rm,m
(1+\sqrt{2})\sqrt{\eps})$-strong extractor.

Let our source have min-entropy $\Hmin{X|E} = \alpha n$.  We want the
entropy loss induced by this one-bit extractor to be roughly $\gamma
n$, and need to find the appropriate \(\ell\) for the desired
value of \(\gamma\) since $\gamma(\ell, \eps) =
h\left(\frac{\ln 2}{\ell} \log \frac{2}{\eps} \right)$ for some
$\gamma < \alpha$. Solving for \(\ell\), we find $\ell(\gamma, \eps) =
\frac{\ln 2}{h^{-1}(\gamma)} \log \frac{2}{\eps}$. This implies that
\(\gamma\) directly influences the length of the seed, which we
discuss below. Since the inverse binary logarithm \(h^{-1}(\cdot)\)
is not analytically available, we need to resort to numerical techniques to
determine the appropriate value of \(\ell\) for a given \(\gamma\).
It is convenient to distinguish the experimental
entropy deficiency \(\alpha\) from the loss induced by the extraction
procedure by introducing a parameter \(\mu\) such that \(\gamma =
\mu\alpha\).

For $\eps = \frac{\eps'^2}{(1+\sqrt{2})^2}$, Trevisan's construction
is a quantum $(k,\eps'm)$-strong extractor with $k = \gamma n + rm + 6
\log \frac{(1+\sqrt{2})}{\eps'} +\log \frac{4}{3}$. The seed of the
one-bit extractor has length $t = \ell \log n = \frac{2 \ln
  2}{h^{-1}(\gamma)} \log n \log \frac{(2+\sqrt{2})}{\eps'}$, and the
seed of the complete construction has length $d$.

Especially the choice of \(\mu\) influences the behaviour of the
XOR extractor. Figures~\ref{fig:xor_params1}, \ref{fig:xor_params2},
and~\ref{fig:xor_params3} depict and discuss the effect of the various
chosen and inferred parameters.

\begin{figure}[htb]
 \centering\includegraphics[width=\linewidth]{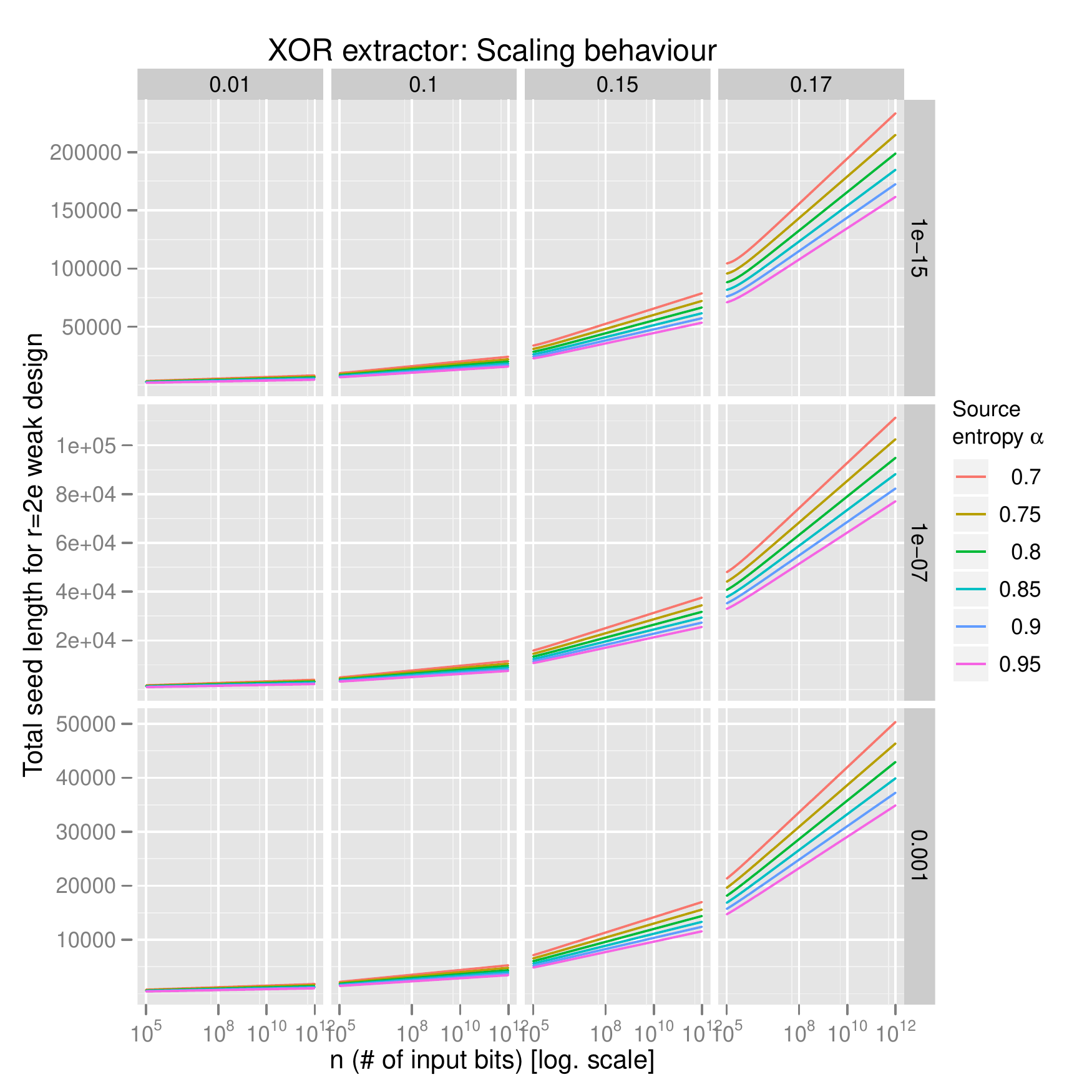}
 \caption{XOR parameter behaviour overview. \(\mu\) varies per column,
   \(\eps\) per row. The influence of the initial source entropy
   \(\alpha\) is mostly negligible, especially for small values of the
   extraction ratio \(\mu\). However, there is a drastic increase in
   seed size for \(\mu > 0.16\), which restricts the XOR method to a
   practical upper bound of the extraction ratio of about 10\% of
   the source entropy.}\label{fig:xor_params1}
\end{figure}

\begin{figure}[htb]
\centering\includegraphics[width=\linewidth]{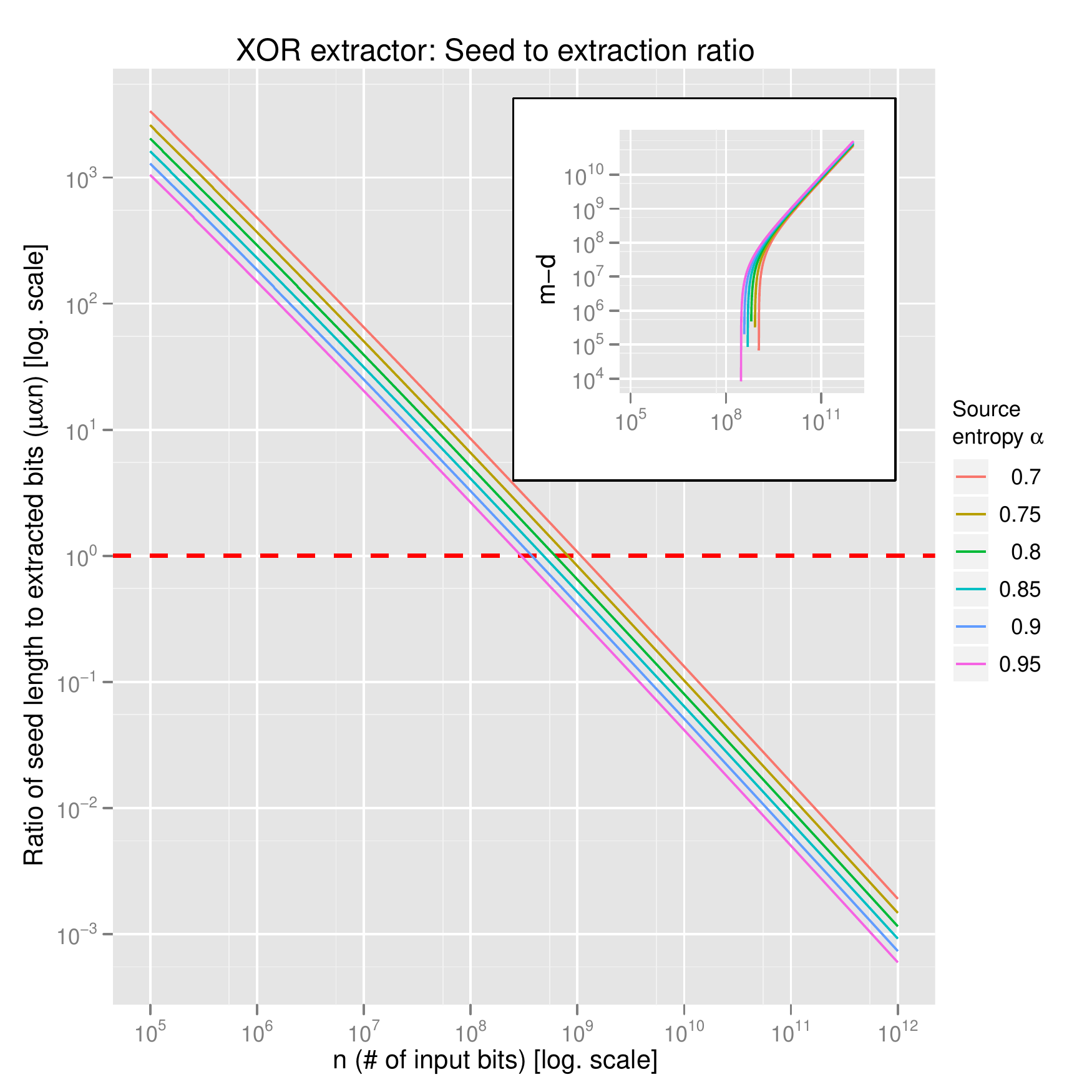}
\caption{Ratio between output size and seed length for various input
  sizes. The parameters \(\mu = 0.05\) and \(\eps=10^{-7}\) are fixed
  for this computation. A ratio of 1, indicated by a dashed red line,
  denotes the parameter regime where the amounts of extracted bits and
  the seed spent coincide; for values exceeding this threshold, the
  particular advantages of Trevisan's construction over two-universal
  hashing prevail because the ratio is better than what can be
  achieved with the latter approach. Recall, though, that the seed
  acts as a catalyst that can be included in the final randomness
  since the extracted bits are independent of the seed.\newline
  The initial source entropy \(\alpha\) accounts for a variation of
  about one order of magnitude of the extraction threshold. As a rule
  of thumb, the break-even point is at input sizes of roughly of
  \(10^{9}\) bits, which amounts to approximately \(2^{30}\) bytes
  (roughly 1~GiB) of data.\newline
  The inset shows the number of extracted bits less the seed
  spent.}\label{fig:xor_params2}
\end{figure}

\begin{figure}[htb]
\centering\includegraphics[width=\linewidth]{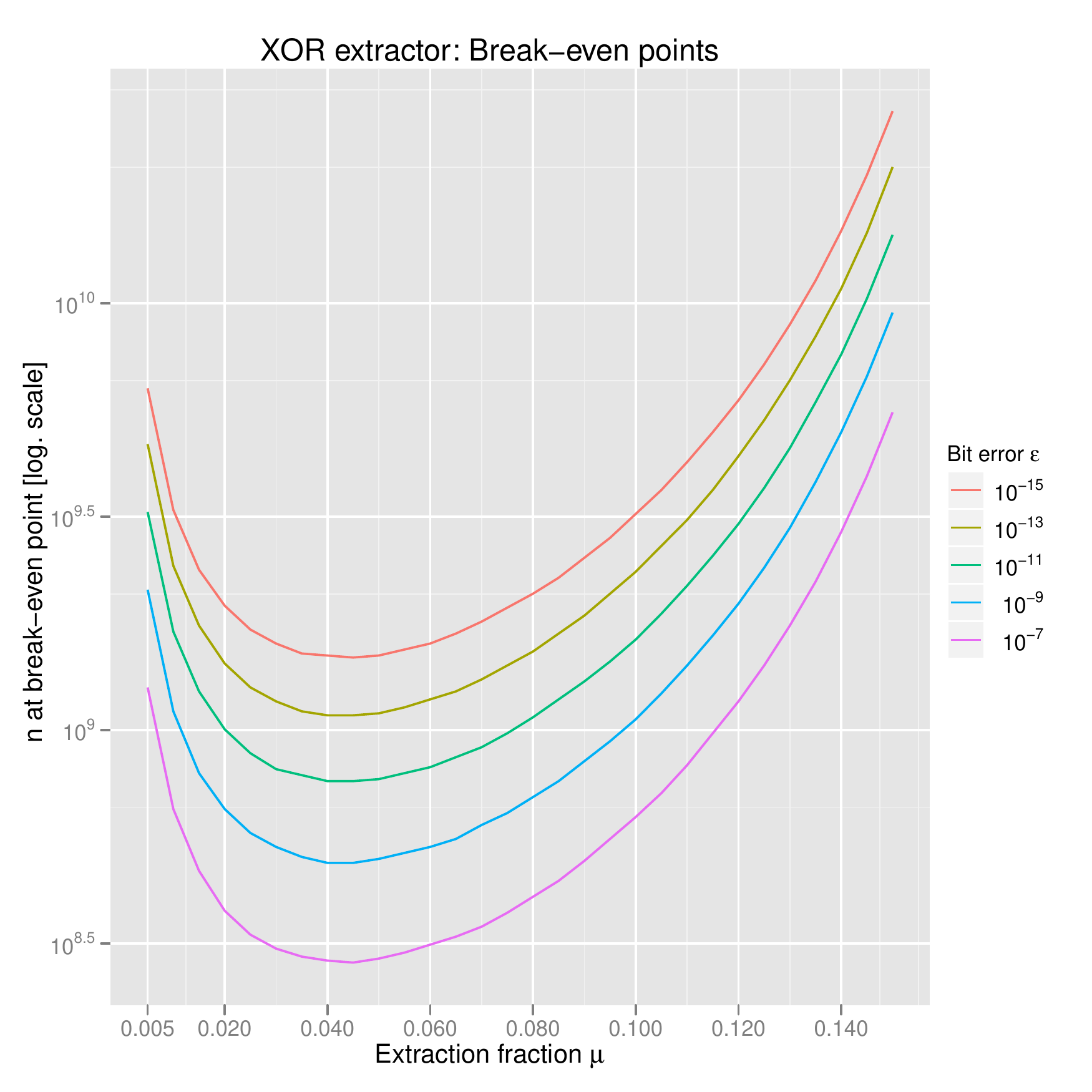}
\caption{Break-even points (\ie, minimal input length for which the
  amount of extracted randomness exceeds the required seed size) for
  varying values of \(\mu\) and \(\varepsilon\). The parameter
  \(\alpha\) is fixed to 0.8. As a rule of thumb, \(\mu = 0.05\) is
  close to the optimal value irregardless of the error parameter
  \(\eps\)).}\label{fig:xor_params3}
\end{figure}

\subsubsection{Lu's construction}\label{sec:lu_expander}

Lu~\cite{Lu2004} shows how to construct a local one-bit extractor,
\ie, an extractor for which each bit of the output only depends on a
subset of the input bits. He then uses his one-bit extractor in Trevisan's
construction. Here, we adapt the parameters of his construction to build a
quantum\-/proof extractor.

Lu's extractor proceeds in two steps. The first consists in selecting
a substring of the input; the second hashes this string to one
bit.\footnote{This type of construction is sometimes referred to as
  \emph{sample-then-extract}~\cite{Vad04}, although Lu~\cite{Lu2004}
  simply describes it as a local list-decodable code.} To select the
substring of the input, he performs a random walk on a $g$-regular
graph \--- a graph in which every vertex is connected to exactly $g$
other vertices.

Recall that a graph \(G\) is uniquely identified by its vertices and
edges, and is consequently specified by \(G = (V, E)\), where \(V\) is
the vertex set and \(E\) the edge set. An alternative representation
of more importance in our context is the \emph{adjacency matrix}.  For
a graph with \(n\) vertices, this is an \(n\times n\) matrix in which
the entry \(a_{ij}\) denotes the number of edges from vertex \(i\) to
vertex \(j\). The diagonal is typically filled with ones; since the
graphs considered here are undirected (\ie, the direction of edges is
not taken into account, only the fact that two vertices are
connected), the adjacency matrix is symmetric.

The eigenvalues of the adjacency matrix are referred to as
\emph{eigenvalues of the graph}.  For our purpose, the ratio between
the second largest and largest eigenvalue plays an important role, and
is labelled as \(\lambda\). Graphs with a small $\lambda$ are called
\emph{expander} graphs, and are common objects in pseudo-randomness
generation, see Ref.~\cite{Goldreich2011} for a review.

For an input string of length $n$, we choose a graph with $n$
vertices, so that each vertex corresponds to a bit position of the
string. Let $(v_1,\dotsc,v_\ell)$ be the vertices visited during a
walk of $\ell$ steps. We select the $\ell$ corresponding bits of the
input $x$, that is, $(x_{v_1},\dotsc,x_{v_\ell})$, and then hash it by
computing the parity of the bitwise product of this string with a random
seed $\beta \in \{0,1\}^\ell$.\footnote{This hash function is also
  used in Section~\ref{sec:derivations.extractors.hashing}.} The
output is thus $z = \bigoplus_{i = 1}^\ell \beta_i x_{v_i}$.

Lu~\cite{Lu2004} proves that the concatenation of the output bits $z$
for all possible seeds is a $(\delta,L)$-list decodable code with
\begin{equation}\label{eq:lucode1}
  L = \frac{2^{h(\nu)n}}{2\delta^{2}},
\end{equation}
for a $\nu \leq 1/2$ given by
\begin{equation}\label{eq:lucode2}
  \nu = 1 + \lambda^2 - \delta^{\frac{4}{\ell}}.
\end{equation}

Since $\delta^{\frac{4}{t}} < 1$, \eqref{eq:lucode2} can only be
satisfied if $\lambda^2 < \nu$. This can be obtained by taking as
expander graph $G$ a given construction $G_0$ to the power $c$. $G$ is
defined as the graph with adjacency matrix $A = A^c_0$, where $A_0$ is
the adjacency matrix of $G_0$. We then have $\lambda = \lambda_0^c$. A
random walk of length $\ell$ on $G_0^c$ is equivalent to a random walk
of length $\ell c$ on $G_{0}$, in which only the first of every $c$ steps
is remembered, and the others deleted~\cite{Hoory2006}.

To construct the regular expander graph $G_0$, we employ an algorithm
reviewed in Ref.~\cite{Goldreich2011}.  Let us only summarise the
essential facts here:

\begin{itemize}
\item The construction is restricted to degree \(g=8\), and the ratio
  between the second-largest and largest
  eigenvalue can be shown to be \(\lambda = 5\sqrt{2}/8 \approx 0.884\).
\item It is possible to compute the graph for all dimensions (\ie,
  number of nodes) that can be expressed as \(\ell^{2}\) for \(\ell
  \in \mathbbm{N}\). This restriction is much more relaxed than for
  other constructions, and does not pose any problems in real
  applications. Formally, the vertex set of the graph is defined on
  \(\mathbbm{Z}_{\ell}\times \mathbbm{Z}_{\ell}\). Each vertex
  \(\vertex{x}{y} \in \mathbbm{Z}_{\ell}\times \mathbbm{Z}_{\ell}\) is
  connected to the vertices \(\vertex{x \pm 2y}{y}\), \(\vertex{x \pm
    (y+1)}{y}\), \(\vertex{x}{y\pm 2x}\), and \(\vertex{x}{y \pm
    (2x+1)}\), which uniquely defines the edges. Notice that the
  arithmetic must be performed modulo \(\ell\), so the computationally
  (comparatively) cheap additions and multiplications are
  unfortunately accompanied by an expensive modulo
  division.\footnote{An obvious optimisation possibility that is
    available because the multiplicative factor 2 is small is to
    compute the modulo division not unconditionally, but only when the
    intermediate result really exceeds \(\ell\).}
\item The complete graph does not need to be computed in advance, but can be
  constructed during the random walk, and using a constant amount of
  space. 
\end{itemize}

For a given $\nu$, we % can fix $\lambda^{2c}_0 = \nu/2$, and thus take
% $c = \left\lceil \frac{\log (\nu/2)}{2 \log \lambda_0} \right\rceil$
% and $\ell = \left\lceil \frac{4 \log \delta}{\log (1-\nu/2)}
% \right\rceil$. Or we can
choose $c$ and $\ell$ which minimize the number
of steps $c\ell$. By setting $w \coloneqq \lambda_0^{2c}$ and taking the
derivative of $c\ell$ with respect to $w$, we find that minimum is
obtained for the $w$ which is the solution of the equation \[w \log w
= (1-\nu+w)\log (1-\nu+w).\] The number of steps on the expander graph
is then given by $c = \left\lceil \frac{\log w}{2 \log \lambda_0}
\right\rceil$ and $\ell = \left\lceil \frac{4 \log \delta}{\log
    (1-\nu+w)} \right\rceil$.

The walk on the $g$-regular graph requires $\nb{n}$ bits of seed to
choose the first vertex, and $c(\ell-1) \log g$ bits for the direction of
the walk for each following step. The final hashing uses $\ell$ bits
of seed, for a total of $t = \nb{n} + c(\ell-1) \log g + \ell$.

From Lemma~\ref{lem:codesRextractors.bis} and \eqref{eq:lucode1}, Lu's
one-bit extractor is a classical-proof $(h(\nu)n + 3 \log
\frac{1}{\delta} - 2,2\delta)$-strong extractor. By
Lemma~\ref{lem:1-bit-against-Q.bis} it is quantum-proof $(h(\nu)n + 3
\log \frac{1}{\delta} - 2,(2+\sqrt{2})\sqrt{\delta})$. And from
Lemma~\ref{lem:implicit}, when used with a weak $(m,t,r,d)$-design,
Trevisan's construction is a quantum-proof $(k,m\eps)$-strong extractor
with $k = h(\nu)n + rm + 6 \log \frac{2+\sqrt{2}}{\eps} - 2$.

Unfortunately, Lu's construction is not useful in a practical setting
owing to its unfortunate parameter scaling: The number of random walk
steps increases considerably with decreasing parameter \(\nu\), see
Figure~\ref{fig:lu_steps}. However, as Figure~\ref{fig:lu_param_nu}
shows, small values of \(\nu\) are required for even tiny extraction
fractions. Overall, this makes the construction reach parameter realms
where it is preferable over two-universal hashing functions (namely,
when the length of the extracted bits exceeds the amount of initial
seed) only rarely, as Figure~\ref{fig:lu_seed_ratio} shows.

\begin{figure}[htb]
 \includegraphics[width=\linewidth]{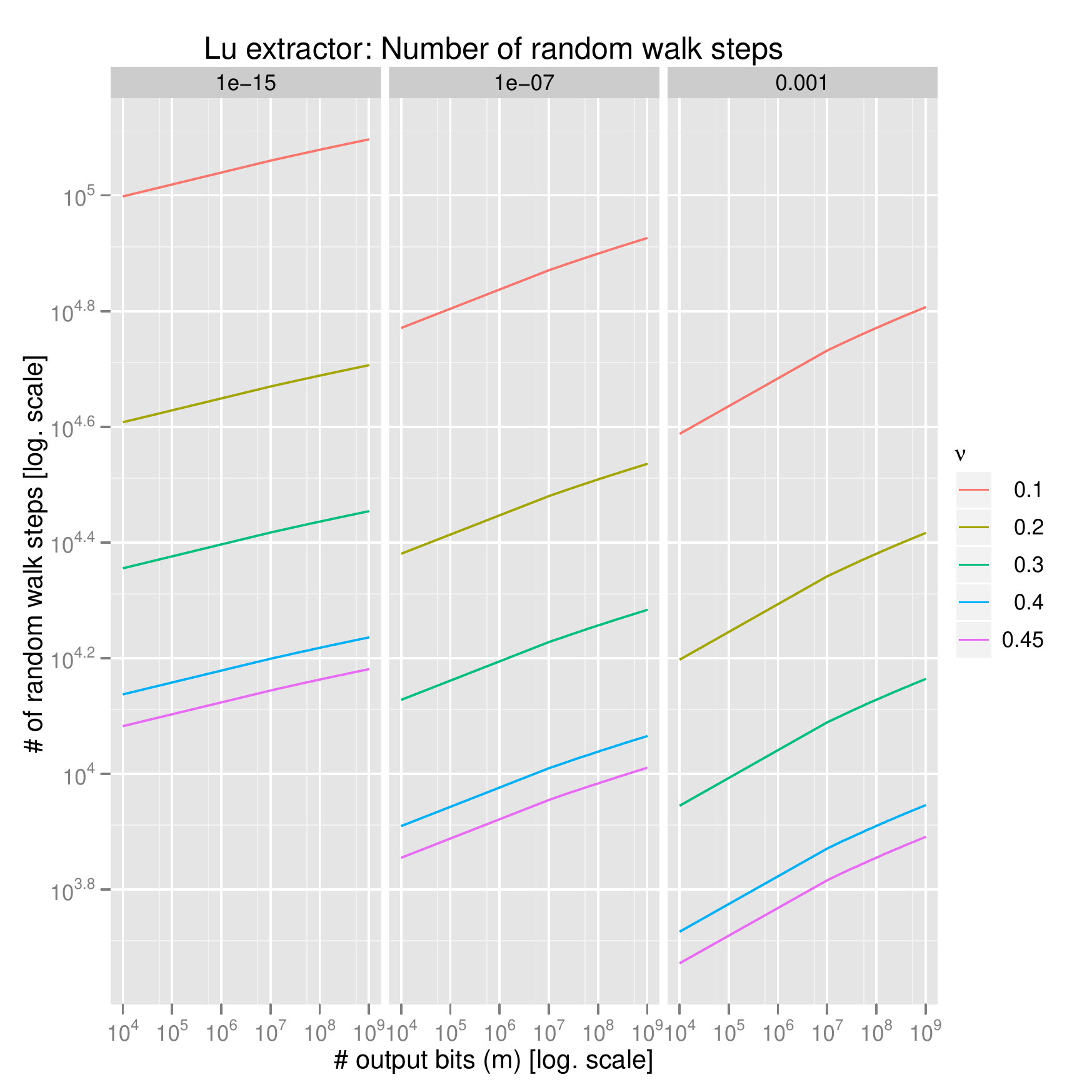}
 \caption{Number of random walk steps required in Lu's construction in
   dependent on the output length and the parameter
   \(\nu\).}\label{fig:lu_steps}
\end{figure}

\begin{figure}[htb]
 \includegraphics[width=\linewidth]{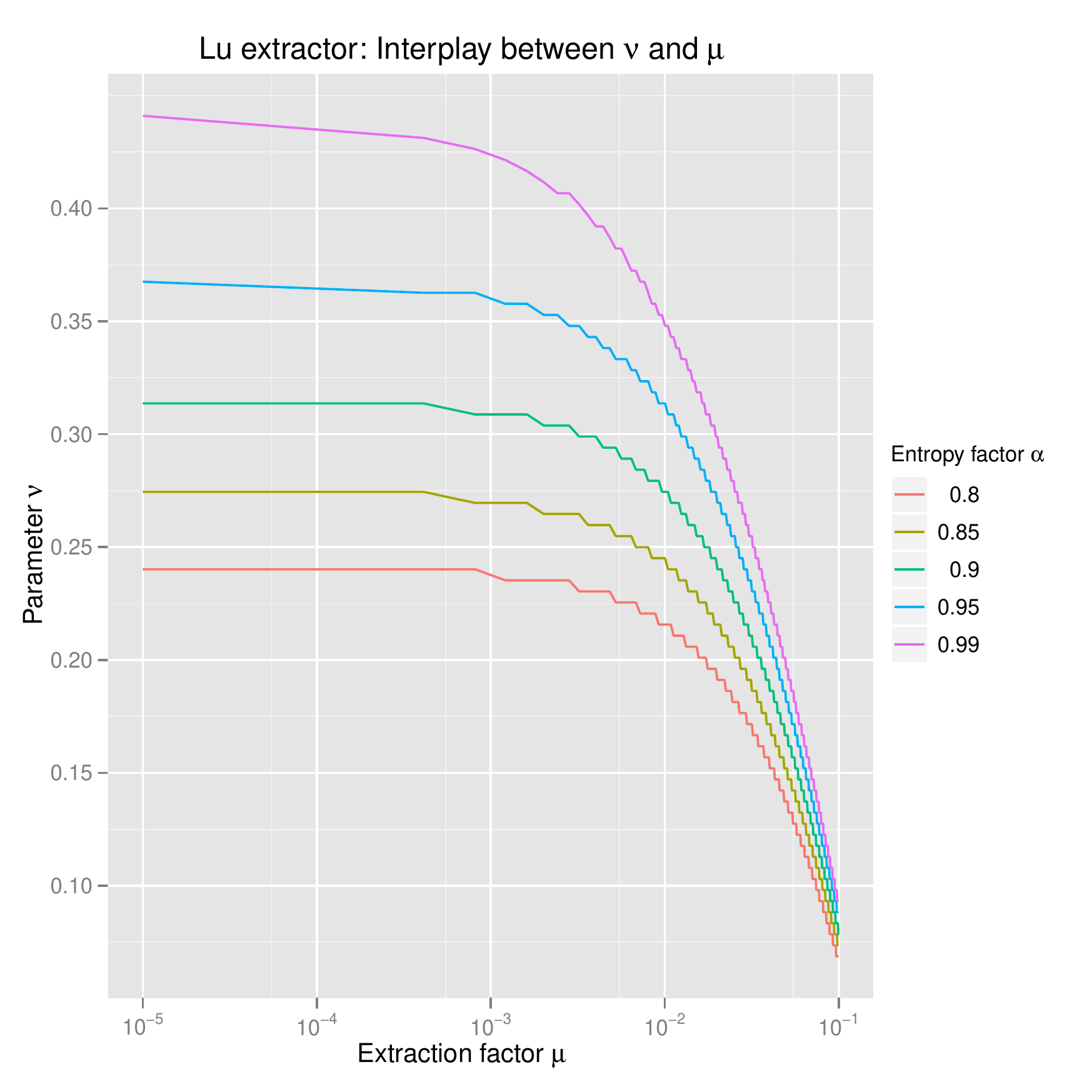}
 \caption{Dependency between parameters \(\nu\) and \(\mu\) for Lu's
   construction. The largest value of \(\nu\) that satisfies the
   boundary conditions on the available entropy given in
   Section~\ref{sec:lu_expander} was determined
   numerically.}\label{fig:lu_param_nu}
\end{figure}

\begin{figure}[htb]
  \includegraphics[width=\linewidth]{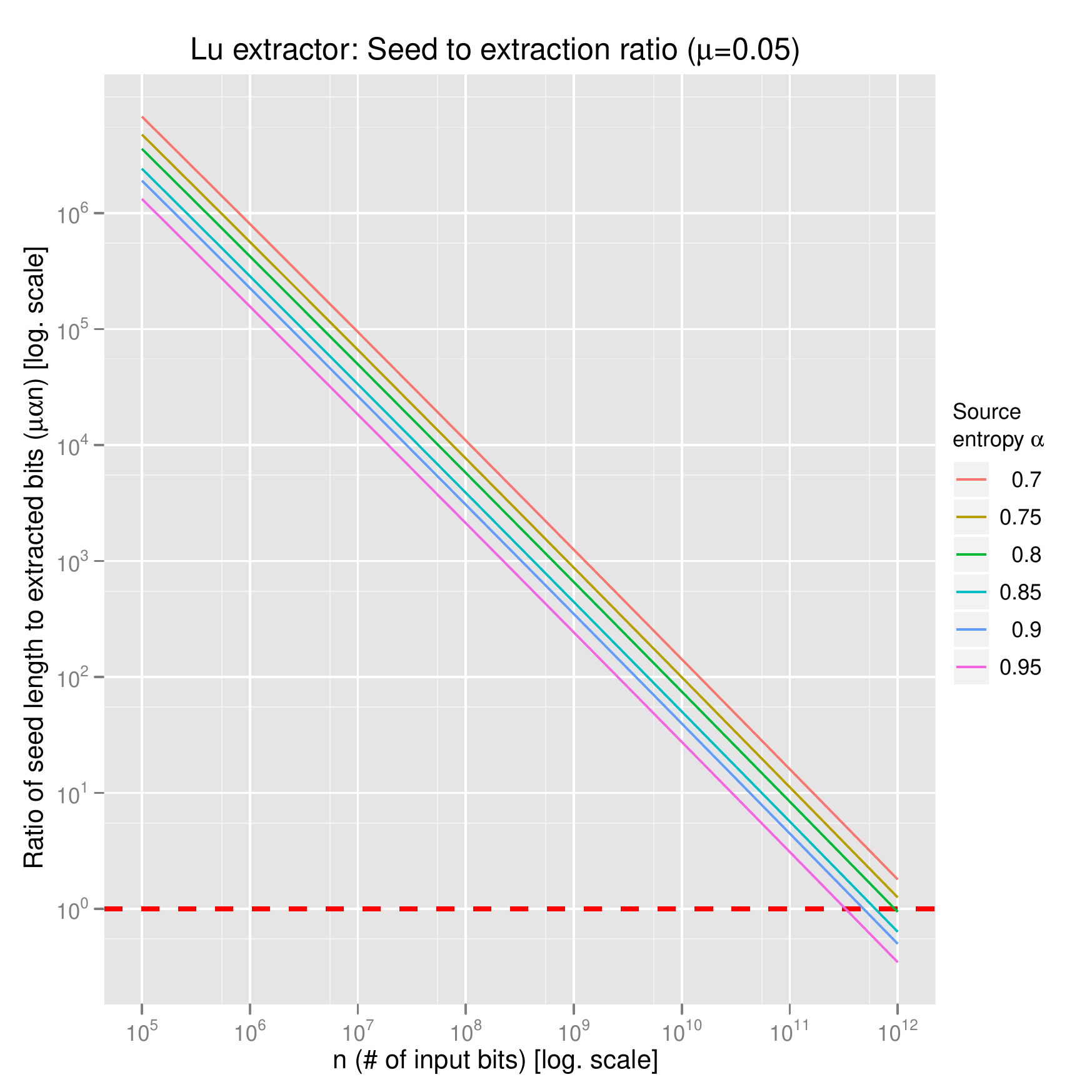}
  \caption{Comparision of seed size and output size for Lu's
    construction. The interesting regime is only reached in very rare
    cases.}\label{fig:lu_seed_ratio}
\end{figure}

\subsubsection{Polynomial hashing}
\label{sec:derivations.extractors.hashing}

Renner~\cite{Ren05} proved that universal$_2$ hash
functions\footnote{See Appendix~\ref{app:known.universal} for a
  definition of (almost) universal$_2$ hashing.} are good
extractors. Tomamichel et al.~\cite{TSSR11} showed that the same holds
for $\delta$-almost universal$_2$ ($\delta$-AU$_2$) hash functions, given that $\delta$ is
small enough. For the range of $\delta$ that build good extractors, almost
universal$_2$ hashing requires a seed of length $\Omega(m+\log n)$,
where $n$ is the input and $m$ the output length. This seed is too
large for many applications; however in the case of one-bit
extractors, this reduces to $\Omega(\log n)$, and is achievable with
the construction we describe here.

This construction is in fact the concatenation of two hash functions,
and uses a seed of length $2\ell$, where $\ell$ will be specified
later. The first is known as \emph{polynomial hashing} \--- or
alternatively as a \emph{Reed-Solomon code}, because the concatenation
of the hashes for all seeds corresponds to the encoding of the input
with a Reed-Solomon code. We partition the input string $x \in
\{0,1\}^n$ in blocks $x = (x_1,\dotsc,x_s)$, each of length $\ell$ (if
necessary, we pad the last string $x_s$ with $0$s). We view each block
as an element of a field $x_i \in \GF{2^\ell}$, and evaluate the
polynomial \[ p_\alpha(x) = \sum_{i=1}^sx_i \alpha^{s-i},\] where
$\alpha \in \GF{2^\ell}$ is the first half of the seed. This family is
$\frac{s-1}{2^\ell}$-AU$_2$.~\cite{Sti95}

Since the $\delta$ of polynomial hashing is too large (relative to the
output length) to build an extractor, we combine it with another hash
function \--- sometimes referred to as a \emph{Hadamard code}, as the
concatenation of the outputs over all seeds corresponds to the
Hadamard encoding. This hash function computes the parity of the
bitwise product of $p_\alpha(x)$ and the second half of the seed,
$\beta \in \{0,1\}^\ell$. The output is thus $z = \bigoplus_{i =
  1}^\ell \beta_i p_\alpha(x)_i$. Since this hash function is
$\frac{1}{2}$-AU$_2$, by \cite[Theorem 5.4]{Sti94} the combination of
the two is $\delta$-AU$_2$ with $\delta =
\frac{1}{2}+\frac{s-1}{2^\ell}$.

Choosing $\ell = \lceil \log n + 2 \log 1/\eps'\rceil$, $s = \lceil
n/\ell\rceil$ we get \[\delta = \frac{1}{2}+\frac{s-1}{2^\ell} <
\frac{1}{2} + \frac{n}{2^\ell} \leq \frac{1}{2} + \eps'^2.\]

From Theorem~\ref{thm:universalhash}, this is a quantum\-/proof
$(4\log\frac{1}{\eps} + 2,2\eps')$-strong extractor. And plugging this
in Trevisan's construction with a $(m,2\ell,r,d)$-design and $\eps =
2\eps'$, we get from Lemma~\ref{lem:implicit} a quantum\-/proof
$(4\log\frac{1}{\eps} + 6 + rm,m\eps)$-strong extractor. The seed of
the one-bit extractor has length $t = 2\ell = 2 \lceil \log n + 2 \log
2/\eps\rceil$, and the seed of the complete construction has length
$d$.

Figure~\ref{fig:rsh_params} discusses the parameters of the polynomial
hashing extractor.

\begin{figure}[htb]
  \centering\includegraphics[width=\linewidth]{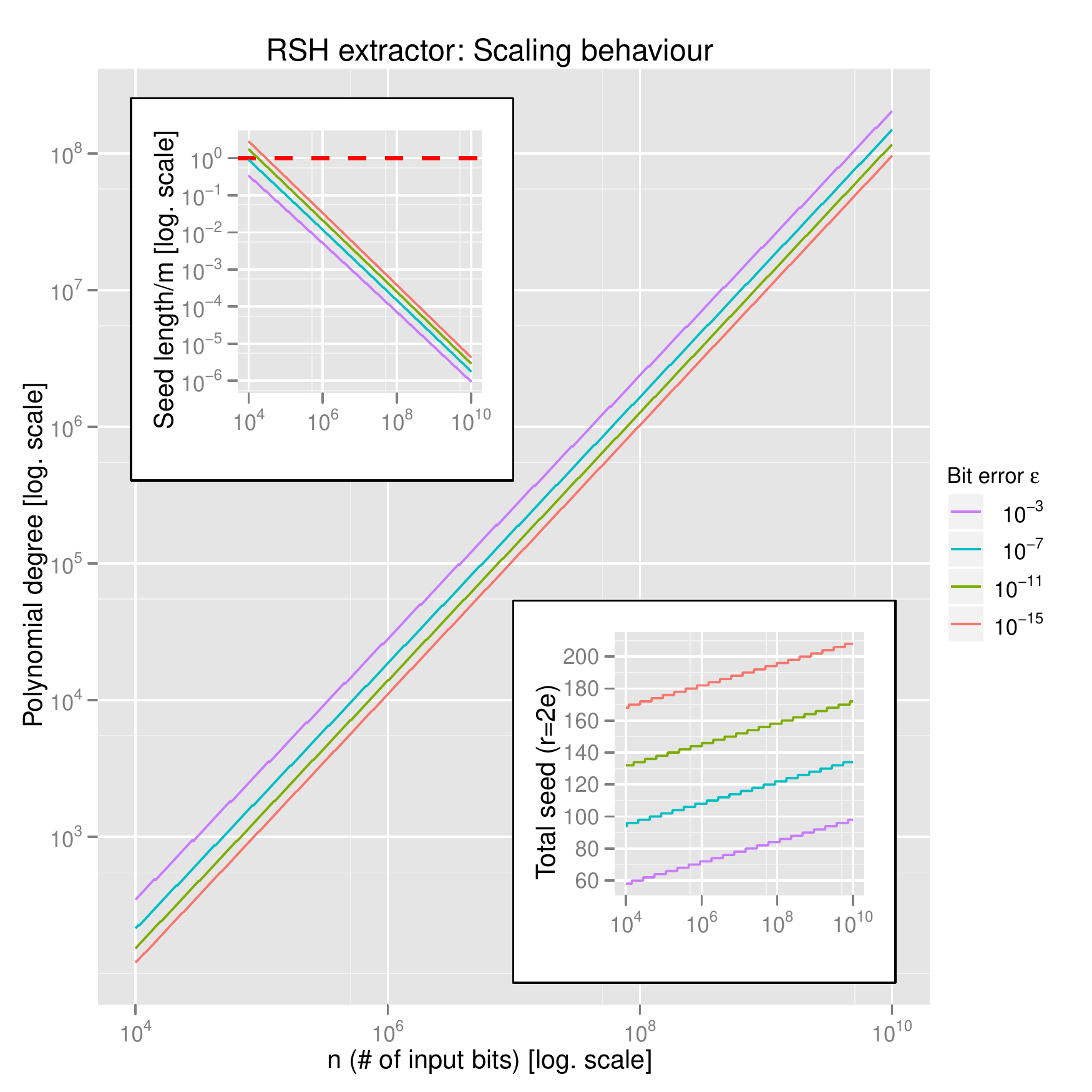}
  \caption{Polynomial hashing parameter overview (calculations are for
    \(r=2e\)). The parameters are easy to evaluate because there is no
    dependence on \(\alpha\), and there is also no extraction factor
    \(\mu\)---the extractor works equally well for high- and
    low-entropy sources. The required seed is consistently small
    (shown in the bottom inset); it increases linearly as
    \(\varepsilon\) decreases exponentially.\newline
   The degree of the polynomial that needs to be evaluated is the
    crucial factor. Even for small inputs like \(n=10^6\),
    corresponding to roughly 1~MiB of data, the degree is
    \(\approx 10000\). Since the polynomial needs to be evaluated for
    every extracted bit, this makes the polynomial hashing extractor
    an unsuitable choice for performance intensive scenarios.\newline
    The top inset shows the regime in which the extractor delivers
    more bits than initially invested for the seed. It outperforms
    two-universal hashing for a very wide range of
    parameters.}.\label{fig:rsh_params}
\end{figure}

% \subsection{Iterative constructions}

% \begin{itemize}
% \item E.g., for running an extractor many times on the same input
% \end{itemize}

%%% Local Variables:
%%% TeX-master: "../trevisan.tex"
%%% End:

\section{Implementation}
\label{sec:implementation}
% Section on implementation issues: Wolfgang
% should include a disucssion on the data types used, libraries used, 
% parallelzation issues, machine types,...
\subsection{Implementation Architecture}
We now turn our attention to describing the implementation of the
Trevisan extraction framework by first outlining the software
architecture, that is, the high-level conceptual point of view,
followed by a discussion of some important implementation details and
notes on how to add new primitives to the infrastructure. While many
important details are still omitted for the sake of brevity, the full
source code is available at
\url{https://github.com/wolfgangmauerer/libtrevisan} for
inspection and modification. Besides instructions on how to build the
code, the website also contains detailled information on how to use
the program, which we will not discuss here any further.

\subsubsection{Architecture}
The architecture was designed to satisfy two particular constraints:
Correctness and maximum throughput. To achieve the latter, we use
C++\footnote{We rely on numerous features of the new language standard
  C++11, so at the time of writing, only sufficiently new compilers
  are able to build the code.} to implement all performance-critical
parts, since the language is statically compiled and does not require
any intermediate layers that add runtime penalties to interpreted or
byte-compiled languages like, for instance, Matlab, but still allows
us to maintain a clean and extensible design based on modern software
engineering techniques~\cite{Stroustrup2000}. The implementation is
portable across a wide range of machines from laptops to high
performance computing (HPC) machines, and also provides opportunities
to benefit from low-level capabilities of recent CPUs, for instance to
accelerate bit-level manipulations. We have tested the code on Linux
and MacOS machines.

To ensure correctness of the calculations, we base the implementation
on independent libraries (NTL~\cite{Shoup2005} and
OpenSSL~\cite{openssl} for working with finite fields of arbitrary
size) that can be selected at compile time.\footnote{NTL cannot be
  used in scenarios with high performance requirements since it is
  restricted to running on one single core per design, which does not
  agree well with contemporary machine architectures. It can only be
  used in a single primitive that requires operations on \(\GF{2^x}\)
  because the library operates with a single, global irreducible
  polynomial, which makes it effectively impossible to operate on
  fields of different dimensions simultaneously.}  Checking that both
variants arrive at the same results for identical parameter sets
increases the faith in the reliability of the calculations.  Another
means to ensure code correctness is given by a large number of
invariants and sanity checks that are spread all across the
implementation. To not compromise the performance goals, it is
possible to deactivate the checks at compile time so that they incur
no runtime penalty.

Another major design decision is the focus on multi-core machines:
Nowadays, machines with only a single core are a rare exception, and
algorithms that are limited to only one thread of execution
voluntarily sacrifice a large fraction of the available computational
power, which is obviously not desirable in a high-performance
setting. We use the threading building blocks
library~\cite{Reinders2007} as basis for the implementation, which
allows for fine-tuning the distribution of work across the system
ressources in a precise manner. We also employ a mostly lock-free
architecture (see, \eg, Ref.~\cite{Herlihy2008} for a review) that avoids
any computation stalls due to the need for synchronised communication
between computation elements.

The code also contains parts that are not performance-critical, for
instance calculating the parameters from given user settings. This is
conveniently done in very high-level languages that allow for working
in abstract terms without having to consider any details of the
underlying machine architecture. To this end, we have integrated the
possibility to call code written in the R language (using the
techniques provided by Ref.~\cite{Eddelbuettel2012}; see~\cite{Rlang}
for an overview about R), which enjoys widespread use in statistical
data processing and machine learning.

It is also possible to compute the weak design ahead of time, store it
on disk and re-use it for multiple runs of the extractor \--- since
computing the weak design is a deterministic operation that does not
require any randomness, this is admissible to do.  In matrix
representation, a weak design for output length \(m\) and a total seed
length \(d\) is an element \(\mathbb{F}_{2}^{m\times d}\). Each row
contains \(t\) ones and \(d-t\) zeroes, so the matrix fill for the
standard design is \(\frac{mt}{m(d-t)} \approx 1/t\). A total seed of
50~KiBit, for instance, amounts to a fill of about 0.5\%, which
exceeds the threshold for typical sparse matrix techniques to pay
off~\cite{Golub1996}. We found the data transfer times from the
underlying block device to be longer than the time required to compute
the weak design on the fly, albeit this may change with the
availability of high-speed storage. For the block weak design, the
situation is more favourable since only the basic design needs to be
stored, and the remaining elements can be reconstructed with very
little computational effort.

Finally, we emphasise that the code can either be used in stand-alone
mode (also including a dry-run mode for parameter estimation), or as a
library as part of a larger project.

\subsubsection{Implementation details}
Weak designs and one-bit extractors are implemented as C++ classes
derived from mixed interface/\hspace{0mm}implementation-type base
classes. Trevisan's algorithm solely operates on the base class
objects using dynamic polymorphism, and does not require any knowledge
about the internal structure of the primitives.

The source code contains full information on how to implement and
integrate new primitives, so we only summarise briefly what methods
need to be provided.

Weak designs need to be derived from \texttt{class weakdes}, and must
implement
\begin{itemize}
\item \texttt{compute\_Si(uint64\_t i, vector indices)} \--- compute
  the \(i\)th index set, and store the results in \texttt{indices}.
\item \texttt{compute\_d()} \--- compute the required amount of initial seed.
\item \texttt{get\_r()} \--- report the overlap \(r\) to the
  higher-level algorithms.
\end{itemize}

Optionally, the function \texttt{set\_params(uint64\_t, uint64\_t m)}
can, but need not be implemented to initialise the parameters required
for all weak designs.

Determining \(d\) from \(t\) seems straightforward, but is accompanied
by constraints \--- the \(\GF{2^x}\) based weak design, for instance,
only works for values of \(t\) that can be represented as a power of
2, so the design typically needs to choose larger values (resulting in
more initial seed) than requested.

One-bit extractors need to be derived from \texttt{class bitext}, and
must implement

\begin{itemize}
\item \texttt{num\_random\_bits()} \--- compute the amount
  \(t\) of initial seed bits required for every extracted bit.
\item \texttt{compute\_k()} \--- determine the minimal source entropy
  required by the extractor for the parameter set under consideration.
\item \texttt{extract(void *initial\_rand)} \--- extract one bit using the
 provided subset of the initial randomness.
\end{itemize}

There are also generic functions to assign global randomness and other
generic parameters to the 1-bit extractor. They can, but need not be
provided by an implementation.

On the lower layers, the implementation was designed to use elementary
machine arithmetic (as opposed to software-based multi-precision
arithmetic) whenever possible; this is an obvious precondition for an
implementation with good performance. In all performance critical
operations, logarithms are not computed using floating point, but with
integer operations since usually only floor or ceiling of the result
is required.

The code uses a fixed-width integer data type with 64 bits to
represent potentially large quantities like the number of input
bits. It is important to note that the width of the index data type
sets an upper bound on the amount of randomness that can be handled by
the code, namely to \(2^{w-3}\) bytes (for \(w=32\) respectively
\(w=64\)), which corresponds to \(2^{w}\) bits (the datum is used as
an \emph{index} into a bit field, and this field need not be
representable by a machine quantity). Since contemporary 64-bit
machines cannot handle more than \(2^{48}\) bytes owing to virtual
address space management limits~\cite{Mauerer2008}, the choice does
not introduce any additional limits. To process large amounts of
randomness (multiple gigabytes), 64 bit machines \emph{and} a 64-bit
kernel running on the machine are required, which the code assumes to
be the default setting.

\subsection{Algorithms}
In the following, we give a concise description of all algorithms
in a form that is helpful for actual implementations---in some
contrast to the previously given descriptions that focus more on
mathematical clarity, we provide recipes in a pseudo-formal language
that is close enough to many contemporary imperative and
object-oriented programming languages, yet still sufficiently abstract
to avoid hiding the algorithmic core behind technical
side-work. Although each algorithm can be captured with very few
statements, we remark that a practical implementation needs to account
for many non-trivial technical issues; our reference implementation
published as a part of this paper comprises about 5000~lines of source
code.

\subsubsection{Trevisan's extractor}
The Trevisan algorithm is independent of the type of weak design and
bit extractor used; only the inferred parameters depend on the
specific properties of the components:

\begin{algorithmic}[1]
\Procedure{Trevisan}{\class{WD}, \class{Ext}, $n, m, \mu, \alpha, \epsilon,
 \varrho^{\text{i}}, \varrho^{\text{d}}$}
%    \Statex \Comment{Parameter preparation}
    \State \(t \gets \mdp{Ext}{InputSize}(n, m, \mu, \alpha, \epsilon)\)
    \State \(d \gets \mdp{WD}{InputSize}(t)\)
    \State \textsc{Reserve space for \(m\) bits in \(\varrho^{\text{o}}\)}
    \State \textsc{Reserve space for \(t\) numbers \(\in [d]\) in \(S\)}
    \Statex
%    \Statex \Comment{Randomness extraction}
    \For{\(i \gets 0, m-1\)}\Comment{\textsc{Data parallel}}
        \State \(S \gets \mdp{WD}{computeS}(i)\)
        \State \(b \gets 0\)
        \For{\(j \gets 0, t-1\)}
            \State \(b_{j} \gets \varrho^{\text{i}}_{S_{j}}\)
           \Comment{\textsc{Indices refer to bits}}
        \EndFor
        \State \(\varrho^{\text{o}}_{j} \gets 
                     \mdp{Ext}{extract}(b, \varrho^{\text{d}}\))
    \EndFor
    \Statex
    \State \return{\varrho^{\text{o}}}
\EndProcedure
\end{algorithmic}

The components \class{WD} and \class{Ext} may impose boundary
conditions on the parameters; for instance, the single-bit seed length
\(t\) must be a power of a prime number for the weak designs
implemented in this paper.

\subsubsection{Weak Designs}
\paragraph{Construction of Hartman and Raz}
The weak design of Hartman and Raz is based on evaluating polynomials
over finite field; recall from Section~\ref{sec:primitives_overview}
that the dimension of the field needs to be a power of a prime
number. We have implemented two variants: One based on the extension
field \(\mathbbm{F} = \GF{2^{x}}\), and one based on the prime field
\(\mathbbm{F} = \GF{p}\). The bit extractors can require arbitrary
values of \(t\) that are not necessarily compatible with the
constraints of the weak design. In this case, \(t\) needs to be
increased to the next possible value \(t'\) that can be provided by
the weak design. Consequently, we need to distinguish between \(t\),
which represents the value that can be provided by the weak design,
and \(\treq\), which is the value originally requested by the bit
extractor. It necessarily holds that \(t \geq \treq\).

The basic algorithm for both finite fields is as follows (indices
in square brackets denote bit selections):

\begin{algorithmic}[1]
\Procedure{HR.ComputeS}{$\mathbbm{F}, i, m, t$}
\State \(c \gets \lceil \frac{\log m}{\log t} - 1\rceil\)
\For{\(j \gets 0, c\)} \Comment{Prepare polynomial coefficients}
	\State \(\alpha_{j} \gets i[j\cdot \nb{t}, j+\nb{t}-1] \mod t\)
\EndFor
\Statex
\For{\(a=0, a < \treq\)}
	\State \(b \gets \sum_{j} \alpha_{j}a^{j}\)
       \State \(S^{a}_{i}[\log(t), 2\cdot\log(t)-1] \gets b\)
        \State \(S^{a}_{i} \gets S^{a}_{i} \cdot a\)
        \State \(S^{a}_{i} \gets S^{a}_{i} \mod |\mathbbm{F}|\)
\EndFor
\Statex
\State \return{S}
\EndProcedure
\end{algorithmic}

For a field of prime dimension \(p\), all calculations are performed
modulo \(p\). Notice, though, that it is not sufficient to simply
divide by \(p\) after any multiplication (or addition/subtraction) has
been performed, because this can easily lead to intermediate results
that exceed the maximal bit width available in hardware. Multiplying
two 40-bit numbers, for instance, can result in an 80-bit value, which
exceeds the word size of 32 and 64 bit machines. A na\"ive solution
could fall back to using arbitrary-precision software arithmetic,
which is unfortunately much slower than native machine hardware
arithmetic. Consequently, we use have made sure to use algorithms that
avoid intermediate overflows and can work with multiplicands of up to
\(61\) bits, which is sufficient for our purposes. See the source code
or Ref.~\cite{Arndt2010} for details.

For the extension field \(\GF{2^m}\), it is not sufficient to perform a
simple division of arithmetic results by a scalar to satisfy the
constraints of the finite field. Instead, all elements of the field
are formally interpreted as polynomials over the binary field, and
arithmetic operations are performed modulo an irreducible polynomial
that needs to be constructed dependent on the field order. It can be
shown (see, \eg, Ref.~\cite{Shoup2005}) that for every field order,
an irreducible polynomial of order 3 or 5 exists, so calculations can
be optimised for these cases.

\paragraph{Block Weak Design}
The block weak design is based on a basic design whose matrix
representation is re-used multiple times as part of the total weak
design---once the matrix representation of the basic design is known,
it is possible to construct the complete design by placing
sub-matrices of the basic design matrix on the diagonal of a larger
matrix. One possible implementation could thus use sparse matrix
techniques to store the basic design in memory, and derive all other
blocks from this representation.

When the basic design is not represented by a matrix, but as vectors
of indices, it is possible to compute the content of
\(W_{\text{B},j}^{k}\) from the basic design row \(W_{\text{B},
  0}^{k}\) by adding \(j\cdot t^{2}\) to all values of the set \(S\)
corresponding to the matrix row. Since it is possible to re-arrange
the rows of \(W\) without changing the properties of the weak design,
we use a suitable permutation (derived from the data in
Eq.~\eqref{eq:r.m}, see the source code for details) of the rows of
\(W\) such that all rows that originate from the same row of the basic
design are adjacent to each other, which allows us to cache calls to
the basic construction.  Since the design is traversed from row to row
in the Trevisan algorithm, the permuted row order minimises calls to
the basic construction.

\begin{algorithmic}[1]
\Procedure{BWD.computeS}{\class{WD}, $i, i^{\text{c}}, S^{\text{c}}, t$}
%    \Require \textsc{Space for \(t\) numbers \(\in [d]\) in \(S\) and \(S^{\text{c}}\)}
    \State \textsc{Infer \(j, k\) from \(i\)}
    \If{\(k \neq i^{\text{c}}\)}
    	\State \(i^{\text{c}} \gets k\)
        \State \(S \gets \mdp{WD}{computeS}(i^{\text{c}})\)
        \Statex
        \For{\(\zeta \gets 0, t-1\)} \Comment \textsc{Fill cache}
        	\State \(S^{\text{c}}_{\zeta} \gets S_{\zeta}\)
       \EndFor
    \Else
  	\For{\(\zeta \gets 0, t-1\)}
            \State \(S_{\zeta} \gets S^{\text{c}}_{\zeta} + j\cdot t^{2}\)
       \EndFor
  \EndIf
  \Statex
  \State \return{S}
\EndProcedure
\end{algorithmic}

\subsubsection{1-Bit extractors}
Finally, we discuss the algorithms used for the 1-bit extractors
implemented as part of this paper.

\paragraph{XOR Code}
An implementation of the XOR code requires to derive the parameter
\(l\) from the experimental parameters; since this can be achieved by
a standard numerical optimisation, we will not discuss a formal
algorithm here, but refer the reader to the source code for the
details. The algorithm itself is compact:

\begin{algorithmic}[1]
\Procedure{XOR.extract}{$\varrho^{\text{i}}, \varrho^{\text{d}}$}
  \State \(r \gets 0\)
  \For{\(i \gets 0, l-1\)}
	  \State \(\zeta \gets \varrho^{\text{i}}[i\cdot \nb{n-1},
          	(i+1)\cdot \nb{n-1}-1]\)
          \State \(r \gets r \oplus \varrho^{\text{g}}[\zeta]\)
  \EndFor
  \State \return{r}
\EndProcedure
\end{algorithmic}

\paragraph{Polynomial Hashing}
The algorithm to perform polynomial hashing based on a concatenation of a
Reed-Solomon and a Hadamard code is as follows:

\begin{algorithmic}[1]
\Procedure{RSH.extract}{$\varrho^{\text{i}}, \varrho^{\text{d}}, n, \eps$}
  \State \(\zeta \gets 0\), \(l \gets \lceil \log n + 2\log 2/\eps\rceil\)
  \State \(s \gets \lceil n/l\rceil\)
  \Statex
  \State \textsc{Pick irreducible polynomial for \(\GF{2^l}\)}
  \For{\(i\gets 0, s-1\)} \Comment{Determine coefficients}
  	\State \(c_{i} \gets \varrho^{\text{g}}[i\cdot l, (i+1)\cdot l-1]\)
  \EndFor
  \Statex
  \State \(\alpha \gets \varrho^{\text{d}}[0:l-1]\)
  \Comment{Reed-Solomon step}
  \State \(r \gets \sum_{i=1}^{s} c_{i}\alpha^{s-i}\)
  \Comment{Computed over \(\GF{2^l}\)}
  \Statex 
  \State \(b \gets 0\) \Comment{Hadamard step}
  \For{\(j \gets 0, l-1\)}
  	\State \(b \gets b \oplus (\varrho^{\text{i}}[l+j]\cdot r[j])\)
  \EndFor
  \Statex
  \State \return{b}
\EndProcedure
\end{algorithmic}

Since the length of the global randomness is considerably exceeds the
bit length of the largest quantity representable with elementary
machine data types in all but the most pathological cases, evaluation
of the polynomial has to be performed using arbitrary precision
software arithmetic.

There are two obvious optimisations: The global randomness does not
change across invocations of the RSH extractor, so it is possible to
compute the coefficients of the polynomial once, and re-use the
results in subsequent evaluations.  In a practical implementation, it
is also more efficient to use Horner's rule for evaluating the
polynomial~\cite{Knuth1997} instead of performing the
straight-forward evaluation shown in the algorithm.

The final parity calculation is not done using single-bit operations
in the actual implementation, but is split into two steps: Firstly,
the logical ``and'' operation is computed block-wise on machine-word
sized blocks. Secondly, the parity operation is built on
special-purpose machine operations (or compiler intrinsics) to count
the number of bits set in the result of the ``and'' operation. The
parity can then be derived by checking if the bit count is even or
odd.

\paragraph{Lu's construction}
The algorithm for Lu's extractor based on a random walk on an expander
graph is as follows (we do not discuss how the optimisations required
to determine the parameters \(c\) and \(l\) are performed; see the
source code for details):

\begin{algorithmic}[1]
\Procedure{LU.extract}{$\varrho^{\text{i}}, \varrho^{\text{d}}, c, l$}
  \State \(v \gets \varrho^{\text{i}}[0:\zeta-1]\) \Comment{Initial node}
  \State \(r \gets 0\), \(b \gets 3\) \Comment{3 bits to represent an
    edge}
  \State \(w \gets \varrho^{\text{i}}[\zeta:\zeta+c(l-1)\cdot b-1]\)
  \State \(s \gets \varrho^{\text{i}}[\zeta+c(l-1)\cdot b:]\)
  \Statex
  \For{\(i \gets 0, c-1\)} 
  	\State \(r \gets r \oplus (\varrho^{\text{d}}[v]\cdot
        	s[i])\)
        \Statex
        \For{j=0, l-2} \Comment{Random walk}
        \State \(e \gets w[(i(l-1)+j)\cdot b, (i(l-1)+j+1)\cdot b-1]\)
        \State \(v \gets \call{next.vertex}{v, e}\)
        \EndFor
  \EndFor
  \State \(r \gets r \oplus (\varrho^{\text{d}}[v]\cdot s[c])\)
  \Statex
  \State \return{r}
\EndProcedure
\end{algorithmic}

\(\zeta\) denotes the number of bits required to store the index of a
node. Function \call{next.vertex}{v, e} computes the value of the
next vertex given the current vertex and the next edge; it is a
straight-forward translation of the calculation rule given earlier in
Section~\ref{sec:lu_expander}.

Most of the implementation complexity for the Lu expander stems from
the need to select subsets of bit strings. To simplify distributing
the initial randomness provided by the weak design into three
components as shown above, the actual implementation assumes that the
contributions start on indices that are evenly divisible by the bit
width of the data type used to represent edges. This simplifies the
implementation, but implies that a slightly larger amount of
randomness than theoretically possible is required, albeit the
increase is only by a negligible additive factor.

\section{Runtime comparison}
\label{sec:runtime}
% Section on runtimes: Wolfgang / Christopher
% should include various pictures, compare runtime of various constructions
% should be self-contained, not requiring section 3,4
Owing to the many aspects---throughput, scalability, weak design
versus extractor performance, parameter ranges, machine
characteristics, among others---involved in determining code performance, and
because of the large number of combinations of primitives, it is
neither possible nor reasonable to present measurements for all
cases (since the full sources are available, measurements for
a particular case of interest can be easily conducted by interested
parties). Instead, we focus on a selection of measurements that
describe cases of typical experimental interest. We use two machines
to run the tests; detailled technical specifications are shown in
Table~\ref{tab:meas_machines}. One machine is a standard Laptop
(MacBook Air) that allows for testing the performance on an average
personal computer, and serves as an apt comparison basis to the
machine used for the measurements in Ref.~\cite{Ma}. The second
machine is a sizeable workstation that gives an indication for the
behaviour in high-performance computing scenarios, or when one is
willing to spend substantial computational effort on the
post-processing, for example in scenarios in which the highest
possible security is the foremost priority.

\begin{table}[htb]
  \begin{tabular}{llllllll}
    Machine & CPU &
    \rotatebox{90}{\# CPUs} & \rotatebox{90}{Cores/CPU} & 
    \rotatebox{90}{Threads/core} & \rotatebox{90}{\(\Sigma\) Threads} &
    \rotatebox{90}{RAM [GiB]} & Kernel \\\hline
    Laptop & Intel Core i5 & 1 & 2 & 2 & 4 &
    4 & Darwin 11.4.0\\[-0.2em]
    & 1.6~Ghz & & & & & &\\[0.2em]
    Workstation & AMD Opteron & 8\footnote{Pairs of two CPUs share one
    socket} & 6 & 
    1 & 48 & 32 & Linux 3.0\\[-0.2em]
    & 1.9~Ghz & & & & & &\\\hline
  \end{tabular}
  \caption{Machines deployed in the benchmark
    measurements.}\label{tab:meas_machines}
\end{table}

The measurement results are shown in
Figures~\ref{fig:scale_n_rsh_block}, \ref{fig:perf_compare_ma},
\ref{fig:xor_gfp_scaling_n}, \ref{fig:xor_rsh_per_core_compare}, and
\ref{fig:block_rsh_cores_scaling}; refer to the captions for a
detailled discussion of the results.

\begin{figure}[htb]
 \includegraphics[width=\linewidth]{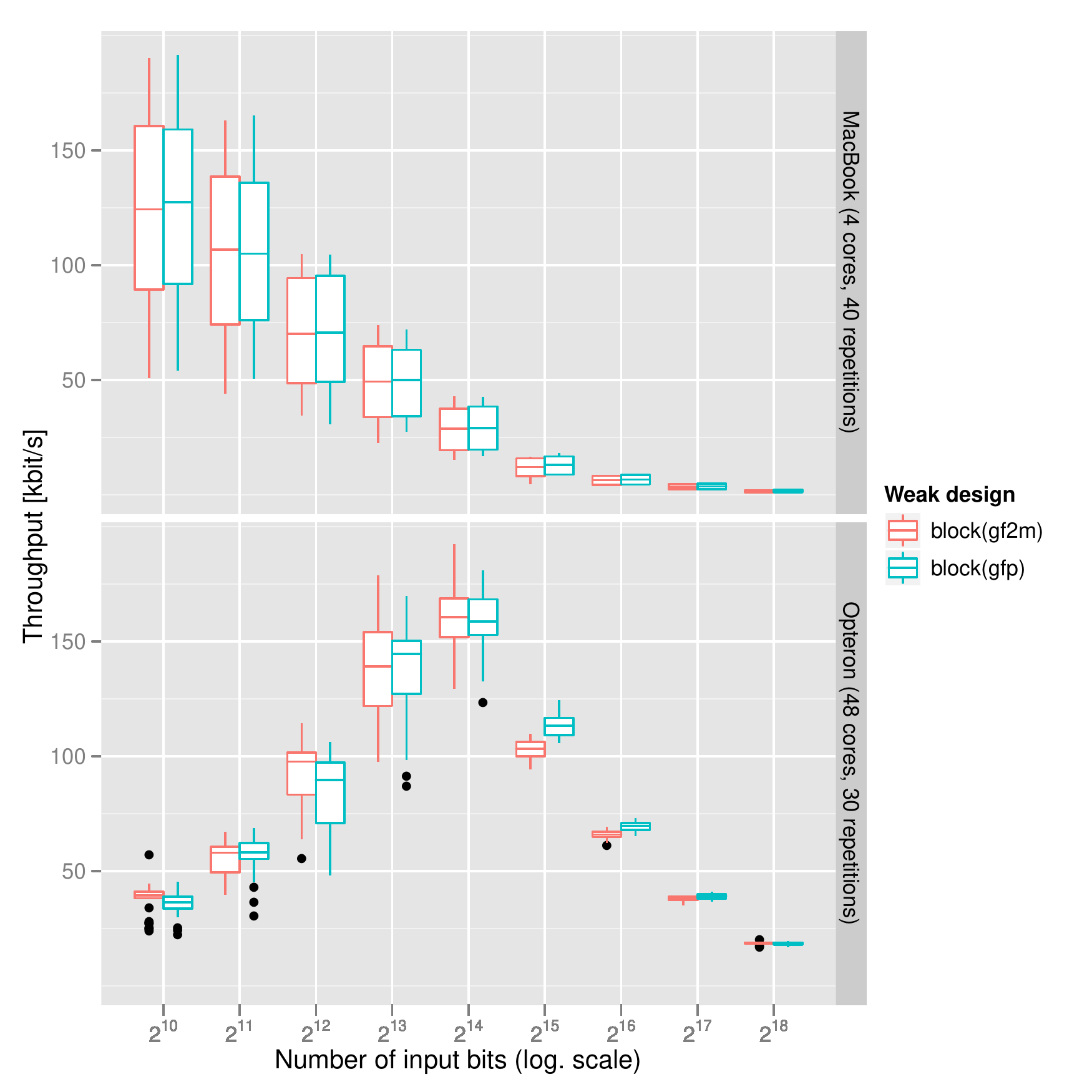} 
 \caption{Scaling behaviour of RSH with a block design for varying
   input lengths. For a small number of CPUs, performance degrades
   considerably with increasing input length, as expected for a
   non-local extactor. Good throughput (more than 100~kbit/s) is only
   obtained fo very small input sizes (\(2^{12}\) is only 4KiBit of
   data!) for which the required amount of initial seed drastically
   exceeds the extracted amount of randomness.\newline
   With many cores, the achieved speed-up does initially \emph{not}
   compensate the overhead for setting up and performing parallel
   operations, so the throughput increases to a local maximum, and
   then decreases as expected with larger input lengths. Consequently,
   it is not just sufficient to add more CPUs for a given scenario to
   increase throughput; practical book-keeping tasks and technical
   aspects can easily dominate the actual problem.  In particular, this
   implies that purely technical improvements like porting the
   processing to massively parallel approaches like GPU computing will
   not automatically resolve all performance needs; a proper choice of
   primitives for given requirements is essential, which is only
   possible with a framework that allows for flexibly combining
   these primitives.}\label{fig:scale_n_rsh_block}
\end{figure}

\begin{figure}[htb]
  \includegraphics[width=\linewidth]{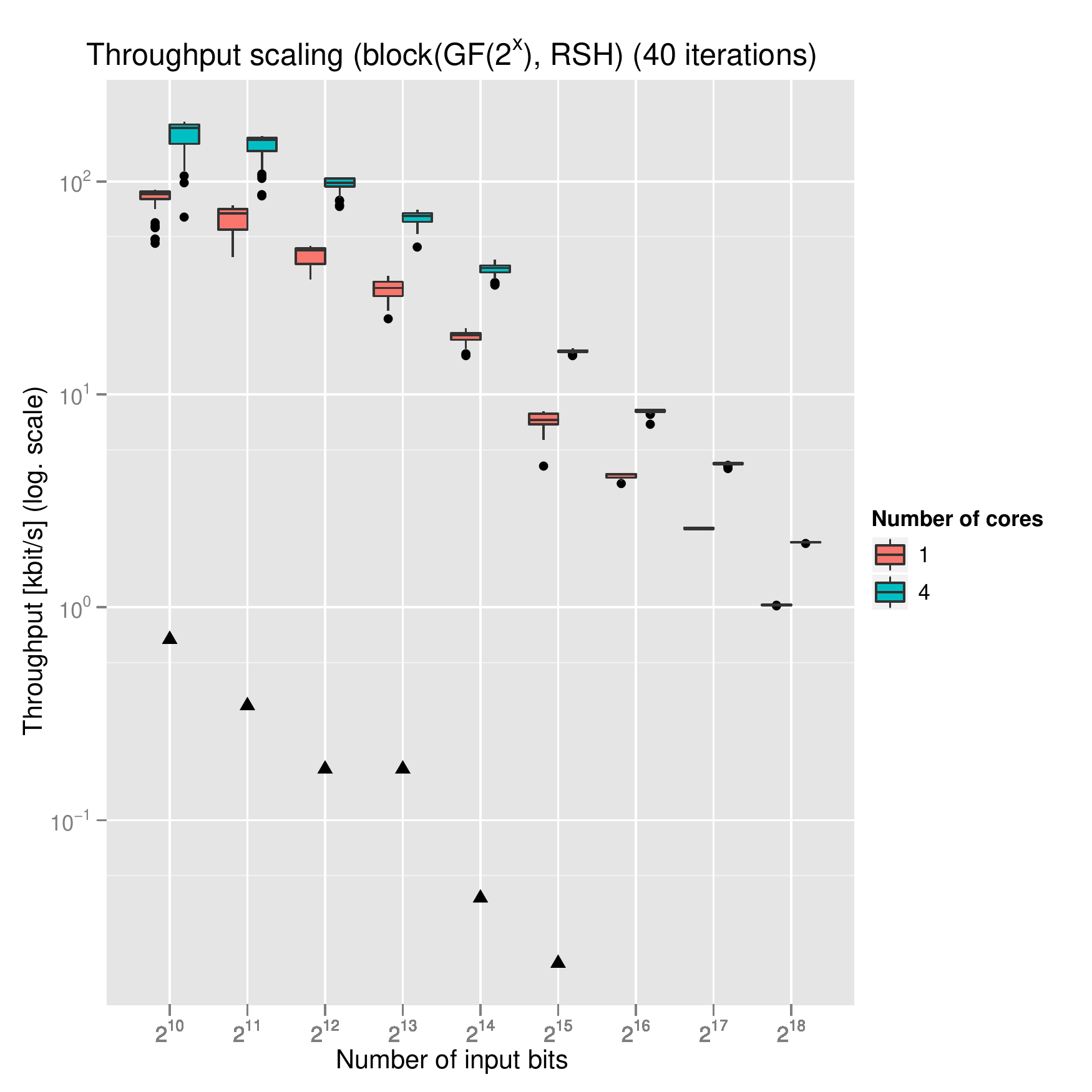}
  \caption{Throughput comparison of our results (obtained on a laptop,
    represented by boxplots) with the results obtained by Ma et
    al.~(represented by triangles) for the combination of primitives
    supported by their implementation. Since the code of~\cite{Ma}
    seems to be limited to running on one CPU core, we have also
    included an artificially contrained measurement measurement for
    the code discussed in this paper. Generally, our framework is 2--3
    orders of magnitude faster in terms of throughput, and allows for
    dealing with inputs that surpass Ref.~\cite{Ma} by many orders of
    magnitude.
%\TODO{Check with Ma how many iterations they used for their code.}
}
  \label{fig:perf_compare_ma}
\end{figure}

\begin{figure}[htb]
  \includegraphics[width=\linewidth]{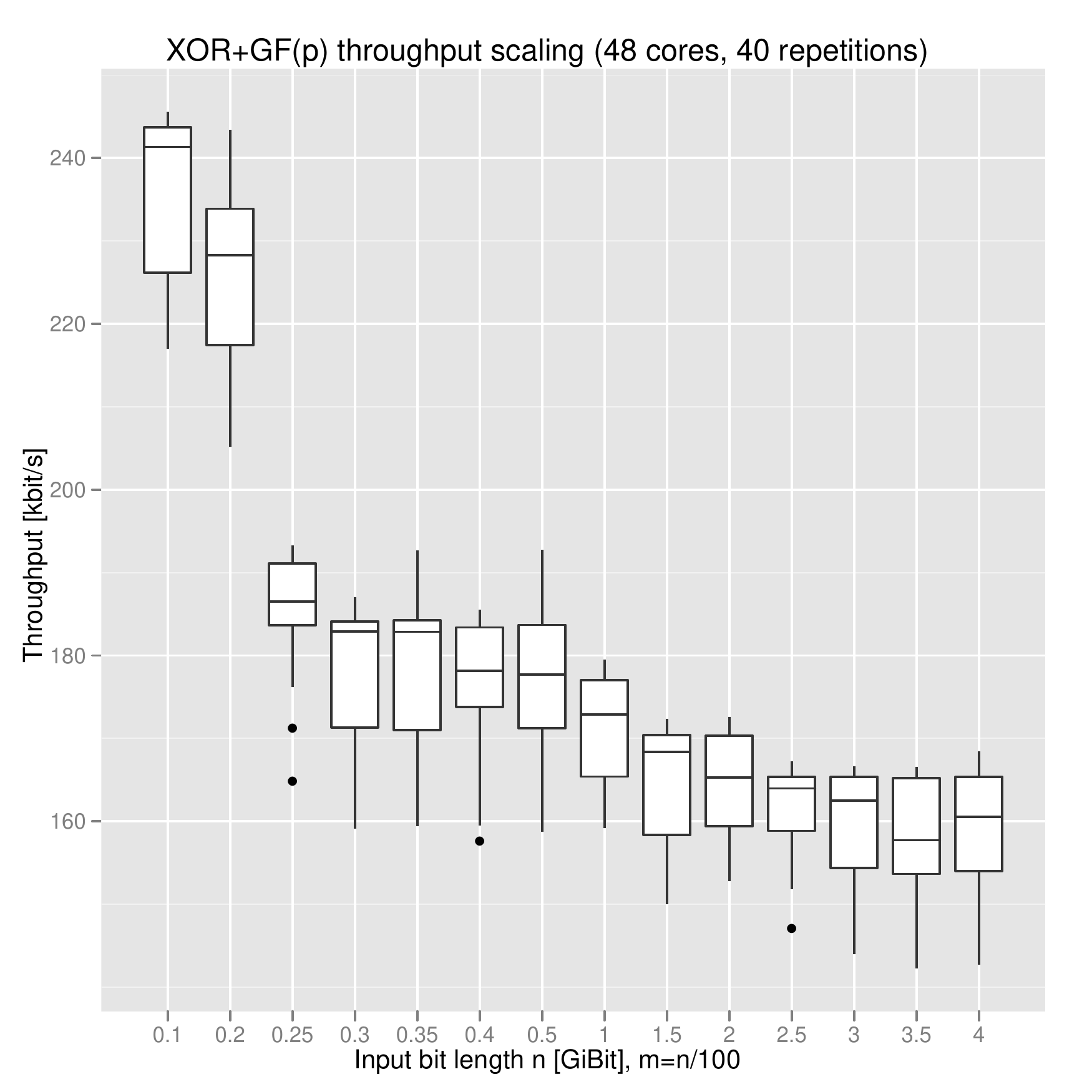}
  \caption{Scaling behaviour of XOR/\(\GF{p}\) for increasing input
    size~\(n\). Although there is a marked decrease in performance for
    input lengths of more than 200~MiBit, a throughput of at least
    160~kbits/s is sustained even for multi-GiBit input lengths (since
    the XOR extracor is local, there is no convergence towards a zero
    throughput rate for long inputs), and matches the requirements of
    typical quantum key distribution mechanisms curently under
    discussion.\newline
    Since only 1\% of the \emph{input} is extracted, the
    code needs to deal with input data rates of 16--20 MiBit/s,
    making the primitives suitable to extract randomness from fast
    random number sources---one example being, for instance,
    Ref.~\cite{Gabriel2010}.}.\label{fig:xor_gfp_scaling_n}
\end{figure}

\begin{figure}[htb]
  \includegraphics[width=\linewidth]{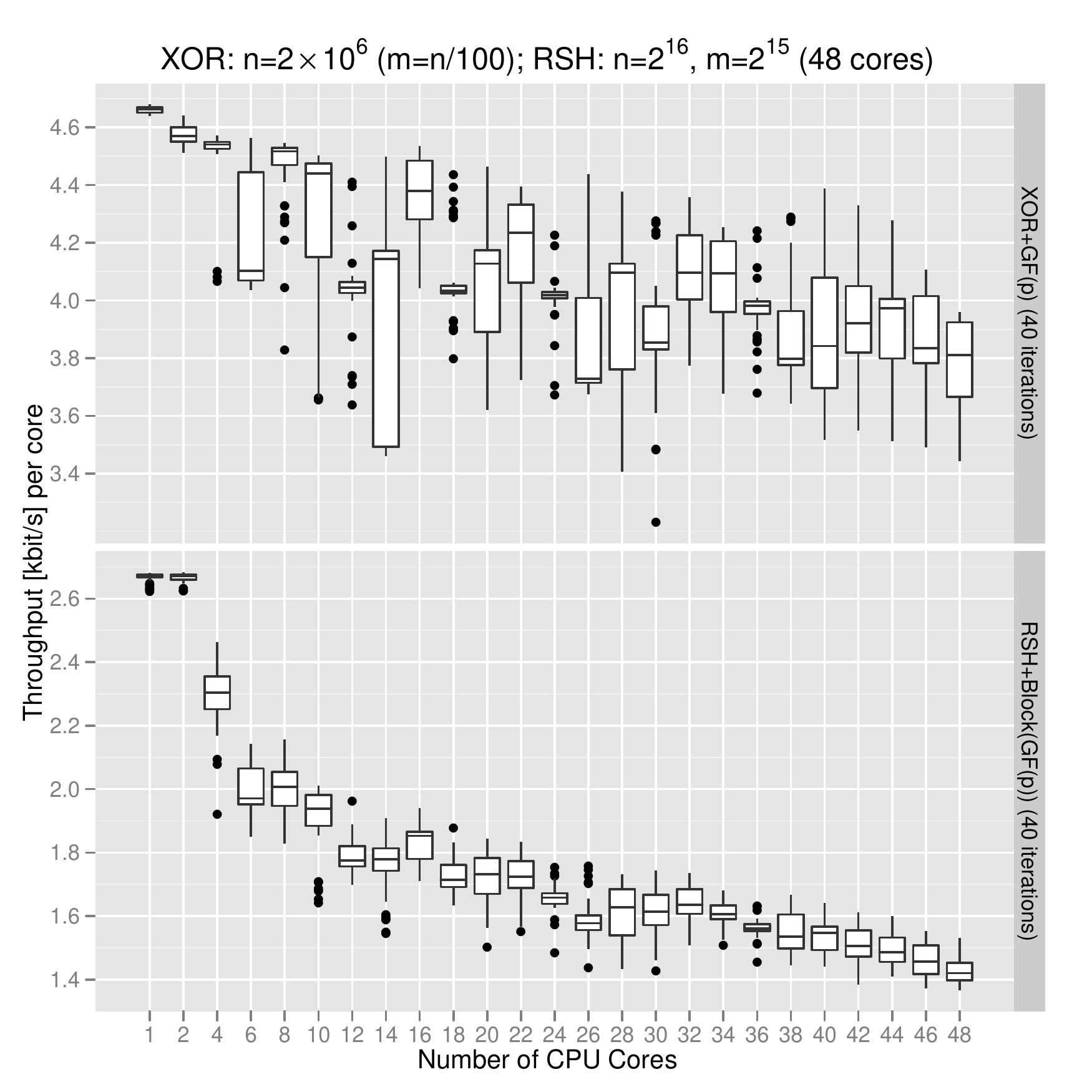}
  \caption{Comparison of per-core performance for XOR/\(\GF{p}\) and
    RSH/Block(\(\GF{p}\)) primitive combinations. The per-core
    throughput for the local XOR extractor drops to about
    75~\% of the sincle-Core performance for a very large
    number of cores (48), which makes it an excellent choice for
    massively parallel systems. For the RSH extractor, performance in
    the many-core case is only half of the performance of a single
    CPU, which can be attributed to the larger amount of data over
    which the primitive combination needs to iterate, and the
    subsequently increased load on the system
    busses.}\label{fig:xor_rsh_per_core_compare}
\end{figure}

\begin{figure}[htb]
  \includegraphics[width=\linewidth]{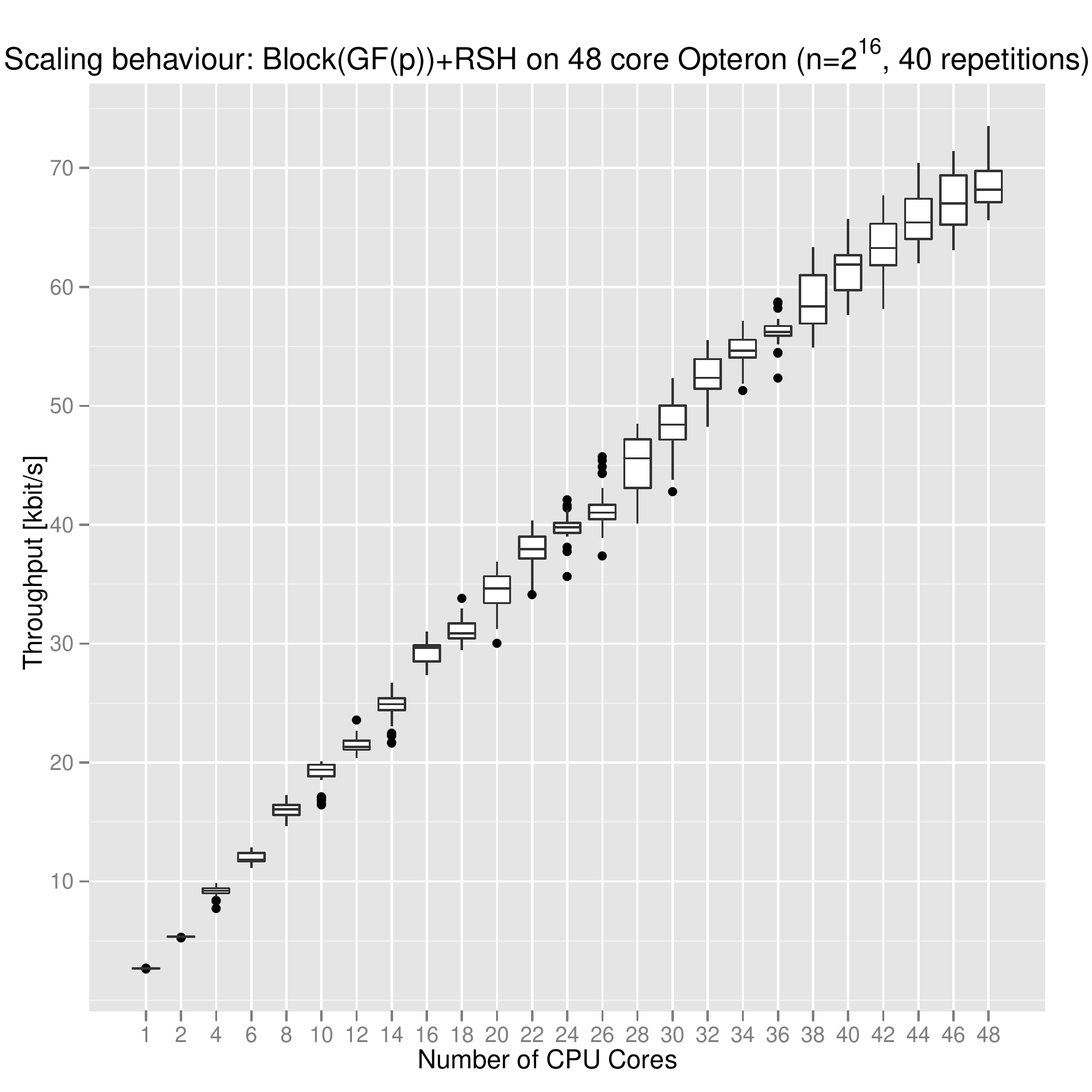}
  \caption{Throughput scaling of the block(\(\GF{p}\)/RSH primitive
    combination for an increasing amount of CPUs. The speedup is well
    below one even for a moderate number of involved cores, and sees a
    further slow-down in the many-core case. Nonetheless, data rates
    that are reasonable for practical application are obtained, making
    the primitive combination a viable choice to post-process the
    output of slow devices for which it is important to not sacrifice
    valuable entropy, as is the case for the faster XOR
    extractor. While the per-core measurement in
    Figure~\ref{fig:xor_rsh_per_core_compare} is more interesting from
    a scalability point of view, this figure provides guidance to what
    ressources are necessary to satisfy given experimental
    constraints.}
  \label{fig:block_rsh_cores_scaling}
\end{figure}

%\begin{figure}[htb]
% \TODO{Re-do this figure for all a relevant subset of the new constructions}
% \caption{Performance measurements to evaluate the contributions of the
%   one-bit extractor and the weak design. For each execution of the primitives
%   (\ie, each time a bit is extracted), the duration of the operation
%   has been measured.}
% \label{fig:perf_wd_vs_bitext}
%\end{figure}

%\TODO{Disuss under which circumstances pre-computing the weak design
%  makes sense, and under which not.} Unfortunately, the cases for
%which pre-computing the weak design could save a substantial amount of
%work coincide with the cases for which the weak design is too large
%for storing it on contemporary machines.

\section*{Summary}
We have presented a modular, scaleable implementation of Trevisan's
construction for randomness extraction, together with detailled
parameter derivations and improved mathematical proofs. We have shown
that the feasibility or non-feasibility of Trevisan's scheme is not
mainly a question of computational complexity issues, but does depend
on the particular choice of primitives used as components of the
algorithm; different scenarios require different
constituents. Although our measurements indicate that there exist use
cases that require theoretical improvements to make Trevisan's
construction applicable (mostly because short-seed extractors all
suffer from a low extraction rate), the implementation can, for
instance, satisfy the needs of all current quantum key distribution
schemes. The authors hope that the public availability of the source
code, together with the extensible architecture, will spawn
contributions from other researchers to turn future theoretical
progress into practical results.

% Acks go here
\section*{Acknowledgements}
WM acknowledges architectural advice on multi-core issues from
T.~Sch{\"u}le, thanks U.~Gleim for a few CPU months, and the ETH
Z{\"u}rich for their hospitality during his stay.

CP is supported by the Swiss National Science Foundation
(via grant No.~200020-135048 and the National Centre of Competence in
Research `Quantum Science and Technology'), and the European Research
Council -- ERC (grant no. 258932).

The authors thank B.~Heim for helpful comments on an earlier draft of
this paper.
\begin{appendix}

\section{Extractor definitions}
\label{app:extractordefs}

An extractor $\Ext: \{0,1\}^n \times \{0,1\}^d \to \{0,1\}^m$ is a
function which takes a weak source of randomness $X$ and a uniformly
random, short seed $Y$, and produces some output $\Ext(X,Y)$, which is
almost uniform. The extractor is said to be strong, if the output is
approximately independent of the seed.

The distance from uniform is measured by the trace distance, defined
as $d(\rho,\sigma) \coloneqq \frac{1}{2} \trnorm{\rho-\sigma}$, where
$\trnorm{\cdot}$ denotes the trace norm given by $\trnorm{A} \coloneqq
\tr \sqrt{\hconj{A}A}$.

\begin{deff}[strong extractor~\cite{NZ96}]
  \label{def:extractor}
  A function $\Ext: \{0,1\}^n \times \{0,1\}^d \to \{0,1\}^m$ is a
  \emph{$(k,\eps)$-strong extractor}, if for all
  distributions $X$ with min-entropy $\Hmin{X} \geq k$ and a uniform seed $Y$, we
  have\footnote{A more standard classical notation would be
    $\frac{1}{2} \left\| \Ext(X,Y) \circ Y - U \circ Y \right\| \leq
    \eps$, where the distance metric is the variational
    distance. However, since classical random variables can be
    represented by quantum states diagonal in the computational basis,
    and the trace distance reduces to the variational distance, we use
    the quantum notation for compatibility with the rest of this
    work.} \[\frac{1}{2} \trnorm{ \rho_{\Ext(X,Y)Y} - \tau_{U}
    \otimes \rho_Y} \leq \eps, \] where $\tau_{U}$ is the fully
  mixed state on a system of dimension $2^m$.
\end{deff}

When (quantum) side information $E$ about the source $X$ is present,
the randomness of the source is measured relative to this side
information. We also require the output of the extractor to be close
to uniform and independent from $E$.

\begin{deff}[quantum\-/proof strong
  extractor~\protect{\cite[Section 2.6]{KR11}}]
  \label{def:quantumextractor}
  A function $\Ext: \{0,1\}^n \times \{0,1\}^d \to \{0,1\}^m$ is a
  \emph{quantum\-/proof} (or simply \emph{quantum})
  \emph{$(k,\eps)$-strong extractor}, if for all
  states $\rho_{XE}$ classical on $X$ with $\Hmin[\rho]{X|E} \geq k$,
  and for a uniform seed $Y$, we have
  \[\frac{1}{2} \trnorm{
    \rho_{\Ext(X,Y)YE} - \tau_{U} \otimes \rho_Y \otimes \rho_E}
  \leq \eps, \]
  where $\tau_{U}$ is the fully mixed state on a
  system of dimension $2^m$.

  The function $\Ext$ is a \emph{classical\-/proof $(k,\eps)$-strong
    extractor with uniform seed} if the same holds with the system $E$
  restricted to classical states.
\end{deff}

Note that any conventional extractor (Definition~\ref{def:extractor})
is classical\-/proof with slightly weaker parameters.

\begin{lem}[\protect{\cite[Section 2.5]{KR11},\cite[Proposition 1]{KT08}}]
  \label{lem:classicalproof}
  Any $(k,\eps)$-strong extractor is a classical\-/proof $(k+\log
  1/\eps,2\eps)$-strong extractor.
\end{lem} 

In the extractor constructions described in Section
\ref{sec:derivations}, we are particularly interested in extractors
which only need to process a few bits of the input for every bit of
output. These extractors are called \emph{local}, and defined as
follows.

\begin{deff}[$\ell$-local extractor \cite{Vad04}]
  \label{def:localextractor}
  An extractor $\Ext: \{0,1\}^n \times \{0,1\}^d \to \{0,1\}^m$ is
  \emph{$\ell$-locally computable} (or \emph{$\ell$-local}), if for
  every $y \in \{0,1\}^d$, the function $x\mapsto \Ext(x,y)$ depends
  on only $\ell$ bits of its input, where the bit locations are
  determined by $y$.
\end{deff}

This notion of local extractors applies equally to extractors with and
without (quantum) side information.

\section{Known extractor results}
\label{app:known}

The next sections contain many known theorems on extractors, which we
need to derive the parameters of the constructions from
Section~\ref{sec:derivations}.

\subsection{List-decodable codes}
\label{app:known.ldc}

A standard error correcting code guarantees that if the error is
small, any string can be uniquely decoded. A list-decodable code
guarantees that for a larger (but bounded) error, any string can be
decoded to a list of possible messages.

\begin{deff}[list-decodable code \cite{Sud00}]
  A code $C : \{0,1\}^n \to \{0,1\}^{\bar{n}}$ is said to be
  $(\eps,L)$-list-decodable if every Hamming ball of relative
  radius $1/2 - \eps$ in $\{0,1\}^{\bar{n}}$ contains at most
  $L$ codewords.
\end{deff}

List-decodable error correcting codes are known to be $1$-bit
extractors~\cite{Lu2004,Vad04}. This has been rewritten out
explicitly in \cite{De2012}.
\begin{lem}[\protect{\cite[Theorem D.3\footnote{In the arXiv version,
      this theorem is numbered C.3}]{De2012}}]
  \label{lem:codesRextractors}
  Let $C : \{0,1\}^n \to \{0,1\}^{\bar{n}}$ be an
  $(\eps,L)$-list-de\-co\-dable code. Then the function 
  \begin{align*}
    C' : \{0,1\}^n \times [\bar{n}] & \to \{0,1\} \\
      (x,y) & \mapsto C(x)_y,
    \end{align*}
  is a $(\log L + \log \frac{1}{2\eps}, 2 \eps)$-strong
  extractor.
\end{lem}

As noted in a footnote of \cite{De2012}, this lemma can be
strengthened to \emph{classical\-/proof} extractors.
\begin{lem}
  \label{lem:codesRextractors.bis}
  Let $C : \{0,1\}^n \to \{0,1\}^{\bar{n}}$ be an
  $(\eps,L)$-list-de\-co\-dable code. Then the function 
  \begin{align*}
    C' : \{0,1\}^n \times [\bar{n}] & \to \{0,1\} \\
      (x,y) & \mapsto C(x)_y,
    \end{align*}
    is a \emph{classical\-/proof} $(\log L + \log \frac{1}{2\eps}, 2
    \eps)$-strong extractor.
\end{lem}

\subsection{One-bit extractors}
\label{app:known.onebit}

K\"onig and Terhal \cite{KT08} show that any one-bit extractor is
quantum\-/proof.
\begin{thm}[\protect{\cite[Theorem III.1]{KT08}}]
  \label{thm:1-bit-against-Q}
  Let $C : \{0,1\}^n \times \{0,1\}^t \to \{0,1\}$ be a
  $(k,\eps)$-strong extractor. Then $C$ is a quantum\-/proof $(k + \log
  1/\eps,3\sqrt{\eps})$-strong extractor.
\end{thm}

If we however have a construction which has already been shown to be a
\emph{classical\-/proof} $(k,\eps)$-strong extractor, then
Theorem~\ref{thm:1-bit-against-Q} can be refined as follows.
\begin{lem}[Implicit in \cite{KT08}]
  \label{lem:1-bit-against-Q.bis}
  Let $C : \{0,1\}^n \times \{0,1\}^t \to \{0,1\}$ be a
  classical\-/proof $(k,\eps)$-strong extractor. Then $C$ is a
  quantum\-/proof $(k,(1+\sqrt{2})\sqrt{\eps})$-strong extractor.
\end{lem}

\subsection{Universal hashing}
\label{app:known.universal}

A family of hash functions is almost universal, if the
probability of a collision is low.
\begin{deff}[\cite{Sti94}] \label{def:universalhash} A family of hash
  functions $\{h : \cX \to \cZ\}$ is said to be $\delta$-almost
  universal$_2$ ($\delta$-AU$_2$), if for any $x,x' \in \cX$ with $x
  \neq x'$, \[ \Pr_h \left[h(x) = h(x') \right] \leq \delta,\] where
  the hash functions are chosen uniformly at random.

The family is said to be universal$_2$, if it is $\delta$-AU$_2$ with
$\delta = \frac{1}{|\cZ|}$.
\end{deff}

Tomamichel et al.~\cite{TSSR11} show that for such a family of hash
functions $\{h_y\}_y$, the corresponding extractor \--- defined as
$\Ext(x,y) \coloneqq h_y(x)$ \--- is quantum\-/proof if $\delta$ is small
enough.
\begin{thm}[\protect{\cite[Theorem 7]{TSSR11}}]
  \label{thm:universalhash}
  If a family of hash functions $\left\{h : \{0,1\}^n \to \{0,1\}^m
  \right\}$ is $\delta$-AU$_2$ for $\delta = \frac{1+2\eps^2}{2^m}$,
  then chosen uniformly at random, they build a quantum\-/proof $(m+4
  \log \frac{1}{\eps} +1,2\eps)$-strong extractor.
\end{thm}

\subsection{Trevisan's extractor}
\label{app:known.trevisan}

In \cite[Theorem 4.6]{De2012}, De et al.\ show that if a $(k,\eps)$-strong one-bit
extractor is used in Trevisan's construction, the final extractor is a
quantum\-/proof $(k+rm+\log 1/\eps,3m\sqrt{\eps})$-strong extractor,
where $m$ is the output length and $r$ is a parameter of the weak
design.

That theorem is the combination of the following implicit lemma and
Lemma~\ref{lem:classicalproof}.

\begin{lem}[Implicit in \cite{De2012}]
  \label{lem:implicit}
  Let $C : \{0,1\}^n \times \{0,1\}^t \to \{0,1\}$ be a
  \emph{quantum\-/proof} $(k,\eps)$-strong extractor with uniform seed
  and $S_1,\dotsc,S_m \subset [d]$ a weak $(m,t,r,d)$-design. Then
  Trevisan's extractor, $\Ext_C : \{0,1\}^n \times \{0,1\}^d \to
  \{0,1\}^m$, is a quantum\-/proof $(k + rm, m\eps)$-strong extractor.
\end{lem}

If we use a one-bit extractor which is known to be quantum proof, we
get better parameters from Lemma~\ref{lem:implicit} than
\cite[Theorem 4.6]{De2012}.

\section{Weak design proofs}
\label{app:weakdesignproofs}

\subsection{Basic construction}
\label{app:weakdesignproofs.basic}

\begin{lem}
  The weak design construction described in
  Section~\ref{sec:derivations.weakdesigns.basic} has $r < 2e$.
\end{lem}

\begin{proof}
Ma and Tan \cite{MT11} prove that if $m \in [t^c,t^{c+1}]$ and $t^c$
divides $m$, then the weak design has $r < e$. The lemma is thus
immediate for $m = k t^c$ and any integer $1 \leq k \leq t$. 

Let $k t^c < m < (k+1) t^{c}$ for some integer $1 \leq k < t$. Since
the construction for $m$ is the same as the construction for $m' =
(k+1) t^{c}$ with the last sets $S_p$ dropped, the overlap can only
decrease. Thus \[ \sum_{q < p} 2^{|S_q \cap S_p|} < em' =
\frac{(k+1)e}{k}kt^c < \frac{k+1}{k} em \leq 2em. \qedhere\]
\end{proof}

\subsection{Reducing the overlap}
\label{app:weakdesignproofs.reducing}

\begin{lem}
  The weak design construction described in
  Section~\ref{sec:derivations.weakdesigns.reducing} has $r = 1$.
\end{lem}

\begin{proof}
  For simplicity, we number the sets of the weak design $W$ with two
  indices $(i,j)$, where $0 \leq i \leq \ell$ and $1 \leq j \leq m_i$,
  and label the corresponding set of the basic weak design $S^i_j$. We
  need to show that the second condition of
  Definition~\ref{def:weakdesign} holds for $r = 1$, namely that for
  all $(i,j)$,
  \[ \sum_{(g,h) < (i,j)} 2^{|S_{g,h} \cap S_{i,j}|} \leq m,\] where
  $\{(g,h) : (g,h) < (i,j) \} \coloneqq \bigcup_{g < i} \{(g,h) : h
  \leq m_g\} \cup \{(i,h) : h < j\}$.

  Note that \eqref{eq:r.m} implies that for all $0 \leq k \leq \ell -
  1$, \begin{equation} \label{eq:proofs.r.1} \sum_{j \leq k} n_j \leq
    \sum_{j \leq k} m_j < \sum_{j \leq k} n_j +1,\end{equation} from
  which we get \begin{equation} \label{eq:proofs.r.2} m_k < \sum_{j
      \leq k} n_j - \sum_{j \leq k-1} m_j + 1 \leq n_k +
    1.\end{equation} Furthermore, from the sum of a geometric series,
  we have \begin{align} \sum_{j \leq k-1} n_j + r'n_k & =
    \frac{1-\left(1-\frac{1}{r'}\right)^k}{1-\left(1-\frac{1}{r'}\right)}n_0
    + r' \left(1-\frac{1}{r'}\right)^k n_0 \notag \\ & = r'n_0 =
    m-r'. \label{eq:proofs.r.3} \end{align}

For any two sets $S_{i,j}$ and $S_{g,h}$ with $i \neq g$, we have
$|S_{g,h} \cap S_{i,j}| = 0$. Thus for any set $S_{i,j}$ with $i \leq
\ell -1$, we have 
\begin{align*} \sum_{(g,h) < (i,j)} 2^{|S_{g,h} \cap S_{i,j}|} & =
  \sum_{g < i, h \leq m_g} 1 + \sum_{h < j} 2^{|S^i_{h} \cap S^i_{j}|}
  \\ & \leq \sum_{g < i}
  m_g + r'm_i \\
  & < \sum_{g < i} n_g + 1+ r'(n_i+1) \\
  & = m+1,\end{align*} where we used \eqref{eq:proofs.r.1} and
\eqref{eq:proofs.r.2} in the second from the last line, and
\eqref{eq:proofs.r.3} in the last line. Since the LHS of
the above inequality is an integer, and the inequality is strict, we
must have \[ \sum_{(g,h) < (i,j)} 2^{|S_{g,h} \cap S_{i,j}|} \leq m.\]

Finally, for the case of $S_{\ell,j}$, note that $\ell$ was chosen
such that $m_\ell \leq t$. This can be seen as follows.
\begin{align*}
m_\ell & = m - \sum_{j \leq \ell-1} m_j \leq m - \sum_{j \leq \ell-1}
n_j \\
    & = m -
    \frac{1-\left(1-\frac{1}{r'}\right)^\ell}{1-\left(1-\frac{1}{r'}\right)}\left(\frac{m}{r'}
      - 1\right) \\
    & = r' + \left(1-\frac{1}{r'}\right)^\ell \left(m - r'\right).
\end{align*}
By plugging \eqref{eq:r.ell} in this, we get $m_\ell \leq t$. Since
$t$ is the size of the finite field, the polynomial used to generate
the elements of $S_{\ell,j}$ has all coefficients $0$, except the
constant term which is $j$. We thus have $S^\ell_j = \{(x,j)\}_{x \in
  \GF{t}}$, and so the sets $\{S_{\ell,j}\}_{j \in \GF{t}}$
have no intersection. Hence
 \[\sum_{(g,h) < (\ell,j)} 2^{|S_{g,h} \cap S_{\ell,j}|} \leq \sum_{g
    \leq \ell} m_\ell = m. \qedhere\]
\end{proof}

%%% Local Variables:
%%% TeX-master: "../trevisan.tex"
%%% End:

\bibliographystyle{apsrev}
\bibliography{trevisan}
\end{appendix}
\end{document}